\documentclass{IEEEtran4PSCC}
\usepackage{cite}
\ifCLASSINFOpdf
   \usepackage[pdftex]{graphicx}
\else
   \usepackage[dvips]{graphicx}
\fi
\usepackage[cmex10]{amsmath}
\usepackage{commath}
\usepackage{mathtools}
\usepackage{steinmetz}		% for phasor
\usepackage{supertabular}	% for long tabular
\usepackage{tabularx}		% for long text in tabular
\usepackage{algorithm}		% for algorithms
\usepackage{algorithmic}		% for algorithms
\usepackage{amsmath,amssymb}
\newtheorem{rem}{Remark}
\newtheorem{defi}{Definition}
\newtheorem{prop}{Proposition}
\newtheorem{thm}{Theorem}
\newtheorem{ass}{Assumption}
\newtheorem{lemma}{Lemma}
\newenvironment{proof}{\vspace{0cm}\paragraph*{Proof}}{\vspace{0cm}\hfill$\blacksquare$}
\usepackage{subfig}
\usepackage{color}
\usepackage{tikz}
\usetikzlibrary{shapes,arrows}
\usepackage{soul}
\usepackage{url} 
\usepackage{afterpage}
\usepackage[makeroom]{cancel}
\usepackage{enumitem} 
\usepackage{changes}
\usepackage{caption}
\newcommand{\normsz}[1]{\lVert#1\rVert}
\newcommand{\distset}[1]{\lvert#1\lvert}
\newcommand{\pfop}[1]{F}

\renewcommand{\r}{\textcolor{black}}
\usepackage{accents}
\newcommand\munderbar[1]{\underaccent{\bar}{#1}}

\makeatletter
\let\old@ps@headings\ps@headings
\let\old@ps@IEEEtitlepagestyle\ps@IEEEtitlepagestyle
\def\psccfooter#1{%
    \def\ps@headings{%
        \old@ps@headings%
        \def\@oddfoot{\strut\hfill#1\hfill\strut}%
        \def\@evenfoot{\strut\hfill#1\hfill\strut}%
    }%
    \def\ps@IEEEtitlepagestyle{%
        \old@ps@IEEEtitlepagestyle%
        \def\@oddfoot{\strut\hfill#1\hfill\strut}%
        \def\@evenfoot{\strut\hfill#1\hfill\strut}%
    }%
    \ps@headings%
}
\makeatother

\begin{document}
%
% paper title
% Titles are generally capitalized except for words such as a, an, and, as,
% at, but, by, for, in, nor, of, on, or, the, to and up, which are usually
% not capitalized unless they are the first or last word of the title.
% Linebreaks \\ can be used within to get better formatting as desired.
% Do not put math or special symbols in the title.
\title{Cross-layer Design for Real-Time Grid Operation: 
\\ Estimation, Optimization and Power Flow}

%% aythors
\author{
\IEEEauthorblockN{Miguel Picallo, Dominic Liao-McPherson, Saverio Bolognani, Florian D{\"o}rfler}
\IEEEauthorblockA{Automatic Control Laboratory, ETH Zurich, 8092 Zurich, Switzerland  \\ \{miguelp,dliaomc,bsaverio,dorfler\}@ethz.ch
}
}

%% To specify the authors when (number of affiliations <= 2)
%\author{
%\IEEEauthorblockN{Author n.1 Name per Affiliation A\\ Author n.2 Name per Affiliation A}
%\IEEEauthorblockA{(Affiliation A) Department Name of Organization \\
%Name of the organization, acronyms acceptable\\
%City, Country\\
%\{email author n.1, email author n.2\}@domain (if desired)}
%\and
%\IEEEauthorblockN{Author n.1 Name per Affiliation B\\ Author n.2 Name per Affiliation B}
%\IEEEauthorblockA{(Affiliation B) Department Name of Organization \\
%Name of the organization, acronyms acceptable\\
%City, Country\\
%\{email author n.1, email author n.2\}@domain (if desired)}
%}

%% To specify the authors when (number of affiliations > 2)
% \author{\IEEEauthorblockN{Author n.1\IEEEauthorrefmark{1},
% Author n.2\IEEEauthorrefmark{2},
% Author n.3\IEEEauthorrefmark{3}, 
% Author n.4\IEEEauthorrefmark{3} and
% Author n.5\IEEEauthorrefmark{4}}
% \IEEEauthorblockA{\IEEEauthorrefmark{1} Department Name of Organization A\\
% Name of the organization A,
% Address A\\ Emails if wanted}
% \IEEEauthorblockA{\IEEEauthorrefmark{2} Department Name of Organization B\\
% Name of the organization B,
% Address B\\ Emails if wanted}
% \IEEEauthorblockA{\IEEEauthorrefmark{3} Department Name of Organization C\\
% Name of the organization C,
% Address C\\ Emails if wanted}
% \IEEEauthorblockA{\IEEEauthorrefmark{4}Department Name of Organization D\\
% Name of the organization D,
% Address D\\ Emails if wanted}
% }

% make the title area
\maketitle

% As a general rule, do not put math, special symbols or citations
% in the abstract
\begin{abstract}
In this paper, we propose a combined Online Feedback Optimization (OFO) and dynamic estimation approach for a real-time power grid operation under time-varying conditions. A dynamic estimation uses grid measurements to generate the information required by an OFO controller, that incrementally steers the controllable power injections set-points towards the solutions of a time-varying AC Optimal Power Flow (AC-OPF) problem. More concretely, we propose a quadratic programming-based OFO that guarantees satisfying the grid operational constraints, like admissible voltage limits. Within the estimation, we design an online power flow solver that efficiently computes power flow approximations in real time. Finally, we certify the stability and convergence of this combined approach under time-varying conditions, and we validate its effectiveness on a simulation with a test feeder and high resolution consumption data.
\end{abstract}
\begin{IEEEkeywords}
AC optimal power flow, dynamic state estimation, online optimization, power flow solver, voltage regulation
\end{IEEEkeywords}

% Use this to place sponsorships
\thanksto{
% Funding by the Swiss Federal Office of Energy through the project “Renewable Management and Real-Time Control Platform (ReMaP)” (SI/501810-01) and the ETH Foundation is gratefully acknowledged.
\noindent Funding by the Swiss Federal Office of Energy through the projects “ReMaP” (SI/501810-01) and “UNICORN” (SI/501708), the Swiss National Science Foundation through NCCR Automation, and the ETH Foundation is gratefully acknowledged.
}

\section{Introduction}

% \begin{itemize}[leftmargin=*]
%     \item Power grid operation undergoing a paradigm shift, especially on distribution grids \cite{}. 
%     \item Online feedback optimization \cite{molzahn2017survey,anese2016optimal,hauswirth2017onlinePF} great performance for operating power grids, demonstrated on distribution grids \cite {ortmann2020experimental}.
%     \item Quadratic Programming solvers, computationally efficient, deployed on embedded systems for real-time optimization \cite{domahidi2013ecos, stellato2020osqp}.
%     \item Distribution grids not enough state measurements. \cite{schenato2014bayesian,picallo2017twostepSE}, need to include estimation in online feedback optimization \cite{picallo2020seopfrt}.
%     \item Contributions:
%     \begin{itemize}
%     \item Introduce QP operator and analyze its convergence under time-varying conditions.
%     \item Use sensitivity-conditioning to include power flow solver into the same time-scale, avoiding nested subiterations.
%     \end{itemize}
% \end{itemize}

% Climate change is motivating a reduction of the power grids carbon footprint. As a result, 
The operation of power systems is experiencing an extensive transformation 
due to the introduction of controllable elements (generation power, curtailment, reactive power in converters, flexible loads, etc.), pervasive sensing (smart meters, phasor measurements units, etc.), and communication networks that enable fast sampling and control-loop rates. These new technologies offer tools to reduce operational costs and improve the efficiency, reliability, and sustainability of existing infrastructure. This is especially relevant in distribution grids, where the share of installed renewable generation is continually increasing, but measurement scarcity, the unpredictability and volatility of household power consumption, and renewable energy integration, makes it challenging to enforce grid specifications and constraints, such as voltage magnitude and line thermal limits. Hence, new control and optimization strategies are essential to unlock the potential of these technologies.

{Grid operation can be optimized} by solving an AC optimal power flow (AC-OPF) problem \cite{molzahn2019surveyrel} to determine controllable power injection set-points that minimize operation costs and satisfy grid operational constraints. However, computing the solution of the AC-OPF in real-time can be challenging, limiting the potential of fast sensor and control-loop rates. Online Feedback Optimization (OFO) \cite{molzahn2017survey, anese2016optimal, hauswirth2017onlinePF, picallo2020seopfrt,tang2017real} is an alternative methodology for operating the grid in real-time, e.g., on sub-second time scales and under rapidly varying conditions, and has shown promising results in both simulations and experimental settings \cite{ortmann2020experimental,ortmann2020fully}. OFO constructs a controller that continuously drives the controllable set-points {towards the minimizer} of a time-varying AC-OPF. {This controller takes incremental optimization steps using grid state measurements as feedback,} which enhances robustness against unmeasured disturbances and model mismatch \cite{doerfler2019autonomousgrids,colombino2019robustness}. Furthermore, OFO enables enforcement of grid specifications, either as soft constraints, i.e., penalizing violations \cite{hauswirth2017onlinePF, picallo2020seopfrt, tang2017real}, or asymptotically via dual approaches \cite{anese2016optimal,ortmann2020experimental, qu2019optimal}. %, or projections \cite{haberle2020non}. 
In grids where measurements are scarce, as is typical in distribution grids, \r{a state estimation \cite{abur2004power} is required to generate the state information used by the OFO controller. Recent work \cite{guo2020solving} proposes to solve a static state estimation problem within each OFO control loop. However, this turns OFO into a nested optimization algorithm, which may jeopardize the real-time applicability of the OFO scheme. Alternatively,} a dynamic state estimation \cite{zhao2019revsdynse} with a linear power flow approximation can be used to generate the state estimates required by the OFO controller \r{without nested iterations}, and preserve overall convergence and stability properties of the interconnected system \cite{picallo2020seopfrt}. However, such linear approximations introduce model mismatch and can lead to performance degradation and constraint violations.

In this paper, we propose an OFO scheme based on quadratic programming, and combine it with a dynamic estimation that includes a nonlinear power flow model. Our contributions are four-fold: First, {we propose a quadratic programming-based OFO inspired by \cite{torrisi2018projected,haberle2020non}. This OFO is able to handle constraints more efficiently than penalization or Lagrangian methods,} and provides guarantees on the maximum constraint violation. Second, we use the recently proposed sensitivity-conditioning approach \cite{picallo2021predictivesensitivity} to design a {Newton's method based} online power flow solver for the dynamic estimation. This solver allows us to include a nonlinear power flow model, and approximate its solutions dynamically in a computationally efficient way that is suitable for real-time system operation. Third, we certify that the interconnection of the OFO, estimation and online power flow solver is stable, convergent, {and quantify the tracking error as a function of the time-varying conditions and the power flow approximation}.
%, and {despite time-varying conditions and a power flow approximation} converges to a neighborhood of the optimal input set-points. 
Fourth and finally, we validate our approach {in simulation} using the IEEE 123-bus test feeder \cite{kersting1991radial} and show its superior performance compared to the gradient based OFO in \cite{picallo2020seopfrt}.

The paper is structured as follows. Section~\ref{sec:grid} presents the grid model and the problem setup. Section~\ref{sec:structure} describes the proposed approach.
%, and Sections~\ref{sec:opt}-\ref{sec:PFsolver} detail its parts. %Section~\ref{sec:opt} details the online feedback optimization part, Section~\ref{sec:estim} the dynamic estimation, and Section~\ref{sec:PFsolver} the online power flow solver. 
Section~\ref{sec:conc} provides convergence guarantees under time-varying conditions. Finally, Section~\ref{sec:test} validates our approach on a simulated test feeder. The proofs of our theoretical results are in the Appendix.

\section{Power System Model and Operation}\label{sec:grid}

{An electrical power grid can be represented by the complex vectors of bus voltages and power injections, $U \in \mathbb{C}^n$ and $S \in \mathbb{C}^n$ respectively, and its admittance matrix $Y \in \mathbb{C}^{n \times n}$.} 
{The state of the grid can be represented using the voltage magnitudes and phases: $x=[\, \abs{U}^T,\angle{U}^T]^T \in \mathbb{R}^{n_x}$. The inputs $u \in \mathcal{U} \subseteq \mathbb{R}^{n_u}$ contain the set-points {of the subset} of controllable active and reactive power injections in $S$ (e.g., renewable generation, power converters or flexible loads), and controllable voltage magnitude levels, e.g., in the slack bus through a tap changer. The exogenous disturbances $d \in  \mathcal{D} \subseteq \mathbb{R}^{n_d}$ are the uncontrollable power injections in $S$, e.g., uncontrollable exogenous loads. The sets $\mathcal{U},\mathcal{D}$ denote the admissible inputs and disturbances installed in the grid. Since the grid power injections cannot be arbitrarily large due to physical limitations and power ratings, we can safely assume that the sets $\mathcal{U},\mathcal{D}$ are compact, i.e., closed and bounded.}

\subsection{Power flow equations and high-voltage solution}
At equilibrium, an electrical grid satisfies the power flow equations: $U \circ \bar{Y} \bar{U} = S$, where $\bar{(\cdot)}$ denotes the complex conjugate and $\circ$ the element-wise product. The power flow equations can be expressed in compact form using a smooth function $h(\cdot): \mathbb{R}^{n_x} \times \mathbb{R}^{n_u} \times \mathbb{R}^{n_d} \to \mathbb{R}^{n_x}$ \cite[Eq.~2.5,2.6]{molzahn2019surveyrel}: 
\begin{equation} \label{eq:implicit_pf}
    h(x,u,d)=0.
\end{equation}
Although the power flow equations may have multiple solutions, in this work we focus on the so-called \textit{high-voltage} solution \cite{dvijotham2017high}. Certain conditions on $\mathcal{U},\mathcal{D}$ and the grid guarantee the existence and uniqueness of the \textit{high-voltage} solution close to the nominal operating point \cite{bolognani2016existence, bernstein2018loadflowexist}. Furthermore, the Jacobian $\nabla_x h(\cdot)$ is invertible at a \textit{high-voltage} solution far from voltage collapse conditions \cite{aolaritei2018voltagestab}. Hence, we assume:
%Assumption~\ref{ass:existsol} is reasonable when operating not overly far from nominal conditions.

\begin{ass}\label{ass:existsol} 
% The sets $\mathcal{U}$ and $\mathcal{D}$ are non-empty and compact, and 
For all $u \in \mathcal{U}$ and $d \in \mathcal{D}$, there exists a unique \textit{high-voltage} solution of $h(x,u,d)=0$, and $\nabla_x h (x,u,d)$ is invertible at that solution.
\end{ass}

Under this assumption, the implicit function theorem \cite{krantz2012implicit} guarantees the existence of a local solution map $\chi(\cdot)$ for the \textit{high-voltage} solution $\chi(u,d)$ at $u,d$, satisfying $h(\chi(u,d),u,d)=0$. For the rest of the paper, we will use $\chi(u,d)$ to denote the \textit{high-voltage} solution of $h(x,u,d)=0$. %The invertibility of $\nabla_x h(x,u,d)$ is widely assumed in these kind of application and is guaranteed in certain cases \cite{chiang1990existence}.

{The invertibility of $\nabla_x h(\cdot)$ allows us to use Newton's method to solve the power flow equations \cite{molzahn2019surveyrel}, and evaluate the solution map $\chi(\cdot)$, which is not available in closed form.} Newton's method computes $\chi(u,d)$ using the recursion
\begin{equation}\label{eq:rootfind}\begin{array}{l}
     x_{k+1} = x_{k} -\nabla_x h(x_{k},u,d)^{-1} h(x_{k},u,d),
\end{array}
\end{equation} 
and converges at a quadratic rate under our assumptions.

\subsection{Grid operation: AC-OPF and measurements}
{In grid operation, the controllable inputs $u$ can be chosen to minimize the operational cost, while satisfying grid operational constraints. This can be represented} as an AC-OPF problem \cite{molzahn2019surveyrel}:
\begin{equation}\label{eq:optprob}\begin{array}{l}
    \min_{u \in [\munderbar{u}_t,\bar{u}_t]} f(u) \text{ s.t. } g(u,d_t) \leq 0,
\end{array}
\end{equation}
where $f(u)$ is the operational cost; $g(u,d_t)= \bar{g}(\chi(u,d_t))\leq 0$ enforces the grid constraints for a disturbance realization $d_t$ at time $t$, e.g., voltage magnitude limits; and $\munderbar{u}_t$ and $\bar{u}_t$ denote the minimum and maximum input available at $t$ within the installed capacity $\mathcal{U}$, i.e., $[\munderbar{u}_t,\bar{u}_t] \subseteq\mathcal{U}$. The AC-OPF \eqref{eq:optprob} is not a static problem, it changes with time due to the temporal variations of the physical grid quantities $\theta_t=[d_t^T,\munderbar{u}_t^T,\bar{u}_t^T]^T \in \Theta = \mathcal{D} \times \mathcal{U} \times \mathcal{U}$ caused by, e.g., a sudden change in the power demand $d$, or available renewable power in $[\munderbar{u},\bar{u}]$. %, where the set $\Theta$ is compact under Assumption~\ref{ass:existsol}. 
The goal in real-time optimal grid operation is to continuously solve \eqref{eq:optprob} as $\theta_t$ varies with time.
% We assume the following to simplify the analysis of \eqref{eq:optprob}:

% % For the rest of the paper, we assume the following:
% \begin{ass}\label{ass:regsol} 
% For all $\theta \in \Theta$, there exists a single stationary point $u^*(\theta)$ of \eqref{eq:optprob}, i.e., satisfying the first-order Karush-Kuhn-Tucker ({KKT}) conditions \cite{bertsekas1997nonlinear}. Furthermore, $u^*(\theta)$ is a global minimizer of \eqref{eq:optprob} and a strongly regular solution \cite[Def.~1.23]{izmailov2014newton}, i.e., $u^*(\theta)$ is Lipschitz continuous.
% \end{ass}

% This assumption is not true in general given the nonconvexity of the power flow equations \cite{bukhsh2013local}, i.e., $g(u,d)$ is not convex in $u$. Yet, since $\Theta$ is compact, it may hold for certain $f(\cdot),g(\cdot)$, e.g., a $\eta$-strongly convex $f(\cdot)$ with a sufficiently large $\eta$. {[add reference]? Or can this assumption be dropped and prove convergence to the set of local minima?}

Standard AC-OPF \eqref{eq:optprob} requires monitoring the grid to obtain full measurements of $d$, unfortunately these are not available in most low-voltage grids. {In such cases, $d$ is estimated by, e.g., solving a state estimation problem \cite{abur2004power}, using} measurements of various physical quantities (like voltages, currents, and power magnitudes and/or phasors) that come from different sensors (e.g., smart meters or phasor measurement units). Furthermore, due to the scarce number of measurements in some grids, these are usually complemented with low-accuracy load forecasts known as \textit{pseudo-measurements} \cite{picallo2017twostepSE}, to achieve numerical observability \cite{baldwin1993power}. At each time $t$, {the measurements $y_t$ depend on the state $x_t$ through a measurement model $\bar{c}(\cdot)$}, i.e.,
\begin{equation}\label{eq:meas}
    y_t=\bar{c}(x_t)+\nu_t,
\end{equation}
where $\nu_t$ is a bounded zero-mean white-noise signal with covariance $\Sigma_{\nu,t}$. %The measurements $y$ allow to generate the estimates $\hat{d}$, e.g., by solving a state estimation problem \cite{abur2004power}, that substitute $d$ in \eqref{eq:optprob}.

\section{Online Feedback Optimization and Estimation}\label{sec:structure}

Due to the complexity of the AC-OPF \eqref{eq:optprob}, solving it in real time may be challenging, and the solutions may quickly become outdated and suboptimal due to rapid variations of $\theta_t$ or model mismatch \cite{doerfler2019autonomousgrids}. Furthermore, the need to estimate $d$ via the measurements $y$ in \eqref{eq:meas} adds an extra layer of computational effort. Therefore, solving \eqref{eq:optprob} in such a \textit{feedforward} fashion may not be suitable for real-time grid operation. {Instead, we propose} a \textit{feedback} architecture combining a measurement-driven dynamic estimation{, e.g., a Kalman filter \cite{jazwinski1970stochfilt}}, and a OFO approach \cite{molzahn2017survey, anese2016optimal, hauswirth2017onlinePF, picallo2020seopfrt}, as illustrated in Fig.~\ref{fig:bldiag}.
This approach allows us to monitor the grid state, and {use it as feedback to} update the input set-points at sub-second sampling and control-loop rates. Moreover, {it improves robustness against model-mismatch and unmeasured disturbances \cite{doerfler2019autonomousgrids,colombino2019robustness}.}

\tikzstyle{block} = [draw, fill=white, rectangle, 
    minimum height=2em, minimum width=3em]
\tikzstyle{block2} = [draw, fill=white, rectangle, 
    minimum height=1.5em, minimum width=3em]
\tikzstyle{sum} = [draw, fill=white, circle, node distance=1cm]
\tikzstyle{input} = [coordinate]
\tikzstyle{output} = [coordinate]
\tikzstyle{pinstyle} = [pin edge={to-,thin,black}]

\begin{figure}[b]
\centering
%\vspace{0.25cm}
\begin{tikzpicture}[auto, node distance=2cm,>=latex']
% upper level
	\node [input] (input0) {};   
    \node [block, right of =input0, node distance = 1cm,align=left] (sat) {Saturation \\ $\munderbar{u} \leq u \leq \bar{u}$};
    \node [block, right of = sat, node distance = 2.8cm,align=left] (sys) {Power System \\ \hspace{0.4cm} $\chi(u,d)$};
    \node [block, right of = sys, node distance = 3.2cm,align=left, minimum height = 1.cm] (meas) {\eqref{eq:meas}: Measurements \\ \hspace{0.15cm} $y=\bar{c}(x)+\nu$};
    
%lower level
    \node [block, below of = sat, node distance = 1.8cm,align=left] (opt) {\eqref{eq:oper}: Online \\ Feedback \\ Optimization};
    \node [block, below of = sys, node distance = 1.55cm,align=left] (pf1) {\eqref{eq:rootfindPSint}: Power \\ Flow solver};
%    \node [block2, below of = sys, node distance = 2.4cm,align=left] (pf2) {PF solver};
    \node [block, below of = meas, node distance = 1.8cm,align=left, minimum height = 1.4cm] (est) {\eqref{eq:KFsimple}: Dynamic \\ estimation of \\ disturbance $d$};
    
    \draw[->] (sat) -- node [name=usat,above,pos=0.46] {} node [name=usat4,above=0,pos=0.5] {$u$} (sys);
    \draw[->] (sys) -- node [name=x,above,pos=0.5] {$x$} (meas);
    \draw[->] (meas) -- node [name=y,left] {$y$} (est);
%    \draw[->] (est) -- (pf2);
%    \draw[->] (pf1) -- node [name=KFeq, above = 3pt] {\eqref{eq:KFsimple}:} (est);
%    \draw[->] (pf2) -- node [name=xdhat,below=-1pt,pos =0.6] {$\hat{x},\hat{d}$} (opt);
%    \draw[->] (opt) -- node [name=u,right,pos=0.6] {$u^+$} (sat);
%    \draw[<->] (pf1) -- node [name=usat2,below] {} (opt);
%    \draw (usat) -- node[name=usat3,right] {$u,\munderbar{u},\bar{u}$} (usat2);
    \draw[->] (usat) |- node[name=usat3,right, pos = 0.29] {$u$} ([yshift=-0.25cm]pf1.north west);
    \draw[->] ([yshift=-0.7cm]pf1.north west) -- node [name=xhat,above=-2pt,pos =0.5] {$\hat{x}$} ([yshift=-0.7cm]opt.north east);
    \draw[->] ([yshift=-0.7cm]est.north west) -- node [name=dhat,above=-2pt,pos =0.5] {$\hat{d}$}  node [name=dhatghost,above=-0pt,pos =0.5] {} ([yshift=-0.7cm]pf1.north east);
    \draw[->] ([yshift=-0.25cm]pf1.north east) -- node [name=xhat2,above=-2pt,pos =0.5] {$\hat{x}$} ([yshift=-0.25cm]est.north west);
    \draw[->] (dhatghost) |- ([yshift=-1.2cm]opt.north east);
    \draw[->] ([xshift=-0.5cm]opt.north) -- node [name=u,right,pos=0.6] {$u^+$} ([xshift=-0.5cm]sat.south);
    \draw[->] ([xshift=0.3cm]sat.south) -- node [name=usatlim,right,pos=0.5] {$\munderbar{u},\bar{u}$} ([xshift=0.3cm]opt.north);

    \node [block, right of =sys, node distance = 0.38cm,fill = red, fill opacity=0.2, minimum height = 1.2cm, minimum width = 8.7cm] (sysblock) {};
    \node [block, below of =sysblock, node distance = 1.8cm,fill = blue, fill opacity=0.2, minimum height = 1.6cm, minimum width = 8.7cm] (algblock) {};
    % \node [block, below of =x, node distance = 2.2cm,fill = green, fill opacity=0.2, minimum height = 1.6cm, minimum width = 5.1cm] (estblock) {};
    \node[text=red, above of = usat, node distance = 0.7cm] (physsys) {\textit{Physical system}};
    \node[text=blue, below of = pf1, node distance =1.3cm] (optest) {\textit{OFO and estimation}};
        
    \node [input,above of = sys, node distance = 1.cm] (inpdist) {};
    \draw[->] (inpdist) -- node[name=dist,right,pos =0.2] {$d$} (sys);    
    \node [input,above of = meas, node distance = 1.cm] (inpnoise) {};
    \draw[->] (inpnoise) -- node[name=noise,right,pos =0.2] {$\nu$} (meas);  
\end{tikzpicture}

%\vspace{0.2cm}
\caption{Real-time power grid (upper red block) operation using Online Feedback optimization (OFO), dynamic estimation and online power flow solvers (lower blue block).}\label{fig:bldiag}
\end{figure}
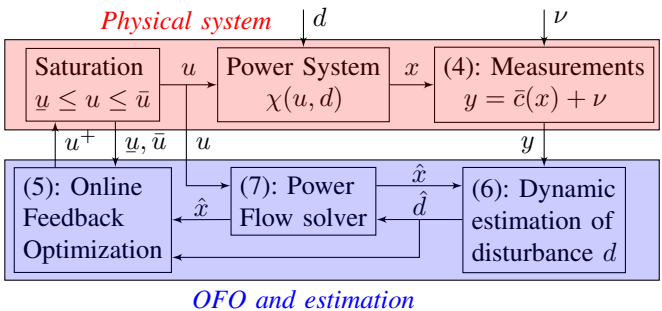

\subsection{Blocks and structure}

\subsubsection{Online Feedback Optimization (OFO)} The OFO controller uses the measured values of $x_t$ and $d_t$ as feedback to incrementally steer $u_t$ towards minimizers of \eqref{eq:optprob} as $\theta_t$ varies in time. The controller is based on an operator $T:\mathcal{U} \times \mathbb{R}^{n_x} \times \Theta \to [\munderbar{u},\bar{u}]$, that enforces the operational limits:
\begin{equation}\label{eq:oper}
    % u_{t+1} = T(u_t ,x_{t},\theta_{t}) \in [\munderbar{u}_{t},\bar{u}_{t}]
    u_{t}^+ = T(u_t ,x_{t},\theta_{t}) % \in [\munderbar{u}_{t},\bar{u}_{t}], DLM: removed since this is redundant with \to [\munderbar ...] above!
\end{equation}
{where $u_t$ and $u_t^+$ denote the inputs before and after the controller update \eqref{eq:oper} at time $t$, respectively. Since the limits $\munderbar{u}_t,\bar{u}_t$ vary with time, the input $u_t^+$ may not be feasible at $t+1$. Hence, the actual input $u_{t+1}$ applied to the grid at time $t+1$ is not $u_t^+$, but rather $u_t^+$ saturated by $[\munderbar{u}_{t+1},\bar{u}_{t+1}]$.}

{Note that the state value $x$ is redundant given $u$ and $d$, as it can be computed as $x=\chi(u,d)$. However, \eqref{eq:oper} takes advantage of both $x$ and $d$ to enhance robustness against model mismatch in $\chi(\cdot)$ \cite{colombino2019robustness} and avoids solving $h(x,u,d)=0$ to save computational time \cite{hauswirth2017onlinePF}. This is particularly relevant in real-time operation, since $T(\cdot)$ must be computationally efficient, especially when using embedded hardware.} 

\r{Note that in our case, see Fig.~\ref{fig:bldiag}, the OFO controller does not have access to the actual values $x_t$ and $d_t$, but instead receives the estimates $\hat{x}_t$ and $\hat{d}_t$ from the estimation block.}

\subsubsection{Dynamic estimation} {It generates estimates $\hat{d}_t$ and $\hat{x}_t$ of $d_t$ and $x_t$ required by OFO \eqref{eq:oper}, see Fig.~\ref{fig:bldiag}. Typically, power system state estimation relies on solving a static nonlinear optimization problem to directly estimate the state $x$ \cite{abur2004power}. However, given the dynamics caused by the temporal variation of $d_t$ and $u_t$, and the nonlinear implicit power flow model $h(\cdot)$, it is preferable to compute a dynamic estimate of $d$ instead of $x$. Furthermore, $d$ is an exogenous quantity unaffected by $u$, and it typically has forecasts available. This dynamic estimation uses the measurements $y_t$ to update the estimates $\hat{d}_t$ via an operator $E:\mathbb{N} \times \mathcal{D}  \times \mathbb{R}^{n_d \times n_d} \times \mathcal{U} \times \mathbb{R}^{n_x} \times \mathbb{R}^{n_y} \to \mathcal{D} \times \mathbb{R}^{n_d \times n_d}$, e.g., a Kalman filter, and $\hat{x}_t$ via the power flow map $\chi(\cdot)$:}
\begin{equation}\label{eq:KFsimple}\begin{array}{rl}
    (\hat{d}_{t},P_t) = & E(t,\hat{d}_{t-1}+\mu_{t-1},P_{t-1},u_t,\chi(u_t ,\hat{d}_{t-1}+\mu_{t-1}),y_t) \\
    % E\big(\hat{d}_{t-1}+\mu_{t},y_t-\bar{c}(\chi(u_t ,\hat{d}_{t-1}+\mu_t))\big) \\ %\hat{d}_{t-1} + \mu_t + K_{t}\big(y_{t} - \bar{c}\big(\chi(u_t ,\hat{d}_{t-1} + \mu_t)\big) \big) \\
    \hat{x}_t = & \chi(u_t ,\hat{d}_{t})
\end{array}
\end{equation}
where $P_t$ is a covariance matrix representing the uncertainty in the estimate $\hat{d}_t$, and $\mu_{t}$ is a drift term that represents the expected temporal variation of $d_t$. %, and $y_t-\bar{c}(\chi(u_t ,\hat{d}_{t-1}+\mu_t))$ is the measurement error given the estimate $\hat{d}_{t-1}+\mu_t$. 

\subsubsection{Online power flow solver} Since a closed-form expression for the \textit{high-voltage} solution map $\chi(\cdot)$ is not available, the dynamic estimation \eqref{eq:KFsimple} requires evaluating $\chi(u_t ,\hat{d}_{t-1} + \mu_t)$, and $\chi(u_t ,\hat{d}_{t})$ using a power flow solver, e.g., Newton's method \eqref{eq:rootfind}, see Fig.~\ref{fig:bldiag}. However, solving \eqref{eq:implicit_pf} at each iteration may be computationally expensive. Instead, we design an efficient online power flow solver {based on the recently proposed sensitivity-conditioning \cite{picallo2021predictivesensitivity}}. This solver generates a power flow approximation $z_t \approx \hat{x}_t$, and updates it at every step using an operator $\pfop{}:\mathbb{R}^{n_x} \times \mathcal{U}^2 \times \mathcal{D}^2 \to \mathbb{R}^{n_x}$:
\begin{equation}\label{eq:rootfindPSint}
    z_{t+1} = \pfop{}(z_{t},u_{t+1},u_{t},\hat{d}_{t+1},\hat{d}_t).
\end{equation}

In the remainder of the paper, we explain how to design these blocks, {i.e., the operators $T,E,\pfop{}$ in \eqref{eq:oper}, \eqref{eq:KFsimple} and \eqref{eq:rootfindPSint},} so that the control architecture in Fig.~\ref{fig:bldiag} is suitable for real-time operation, and admits stability and convergence guarantees. Concretely, in Proposition~\ref{prop:hconv} and Theorems~\ref{thm:opconvloc} and \ref{thm:opconvglb} we certify that under time-varying conditions the online power flow approximation $z_t$ converges arbitrarily close to $\hat{x}_t$, and the input set-points $u_t$ converge close to the minimizers $u^*(\hat{\theta}_t)$ of the AC-OPF \eqref{eq:optprob} given $\hat{\theta}_t=[\hat{d}_t^T,\munderbar{u}_t^T,\bar{u}_t^T]^T$, {with a tracking error determined by the power flow approximation error $z_t-\hat{x}_t$ and the temporal variation $\hat{\theta}_{t+1}-\hat{\theta}_t$.}
%In particular, we will provide an online power flow solver to compute $\chi(\cdot)$ simultaneously with \eqref{eq:KFsimple} and \eqref{eq:oper}.
% \begin{equation*}\begin{array}{l}
%     \lim_{t \to 0} \normsz{\hat{x}_t-z_t}_2 \leq \epsilon \\
%     \lim_{t \to 0} \normsz{u_t \hspace{-0.03cm} - \hspace{-0.03cm} u^*(\hat{\theta}_t)}_2 \leq  \tfrac{L}{1-\xi}\sup_{t}(\normsz{\hat{\theta}_{t+1} \hspace{-0.05cm} - \hspace{-0.05cm} \hat{\theta}_{t}}_2 \hspace{-0.03cm} + \hspace{-0.03cm} \normsz{\hat{x}_t \hspace{-0.05cm} - \hspace{-0.05cm} z_t}_2),
% \end{array}
% \end{equation*}%\comment{keep eqs?}
% for some $\epsilon \ll 1,L>0$.

\subsection{Online Feedback Optimization (OFO) block}\label{sec:opt}

{Our objective is to use feedback to track the time-varying solution of the AC-OPF \eqref{eq:optprob}. Hence, the operator $T(\cdot)$ in OFO \eqref{eq:oper} is typically based on an optimization algorithm that has been modified to incorporate measurements. These operators can be classified based on how they handle the grid constraints:
\subsubsection{Soft constraints} Instead of strictly enforcing the constraints $\bar{g}(\chi(u,d))=g(u,d) \leq 0$ as in \eqref{eq:optprob}, some approaches \cite{hauswirth2017onlinePF,tang2017real,picallo2020seopfrt} only penalize violations. The constraints are added to the cost function using a penalty function with a large parameter $\zeta \gg 0$ that discourages constraint violations:
\begin{equation}\label{eq:projgrad}
\bar{\phi}(u,x)=f(u) + \frac{\zeta}{2} \sum_i \max(\bar{g}_i(x),0)^2,    
\end{equation}
and projected gradient descent is used as the operator $T(\cdot)$ \cite{hauswirth2017onlinePF}.} 
\subsubsection{Hard constraints} Primal-dual gradient methods \cite{anese2016optimal, qu2019optimal} or dual ascent methods \cite{bolognani2015reactivePF,ortmann2020experimental} enforce $g(u,d)\leq 0$ asymptotically. However, these methods are harder to tune and tend to result in oscillatory closed-loop behavior \cite{haberle2020non}. Alternatively, quadratic-programming-based operators \cite{torrisi2018projected,haberle2020non,liao2020time}, similar to sequential quadratic programming \cite[Ch.~18]{nocedal2006numerical}, offer bounds on possible constraint violations \cite[Lemma.~4]{haberle2020non}. These operators use first-order approximations of $f(\cdot)$ and $g(\cdot)$, and a regularizing quadratic term based on a continuous positive definite matrix $B(u)\succ 0$. The resulting operator is
%\comment{change to z?}
\begin{equation}\label{eq:qpaproxhes}\begin{array}{l}
T(u,x,\theta) = \\[0.1cm]
\begin{array}{l}
     \underset{u^+ \in [\munderbar{u},\bar{u}]}{\arg\min} 
     \nabla_u f(u)^T(u^+ \hspace{-0.1cm} - \hspace{-0.05cm} u) + \frac{1}{2}(u^+ \hspace{-0.1cm} - \hspace{-0.05cm} u)^T B(u) (u^+ \hspace{-0.1cm} - \hspace{-0.05cm} u) \\[0.1cm]
    \text{s.t. } %\bar{g}(x) + \nabla_x \bar{g}(x) \nabla_u \chi(u,d)(u^+ \hspace{-0.1cm} - \hspace{-0.05cm} u) \leq 0,
    \underbrace{\bar{g}(x) \hspace{-0.05cm} - \hspace{-0.1cm} \nabla_x \bar{g}(x) \nabla_x h(x,u,d)^{-1} \hspace{-0.05cm} \nabla_u h(x,u,d)(u^+ \hspace{-0.1cm} - \hspace{-0.05cm} u)}_{= g(u,d) + \nabla_u g(u,d) %\underbrace{\nabla_x \bar{g}(\chi(u,d)) \nabla_u \chi(u,d)}_{} 
    (u^+  -  u)  \text{ at } x=\chi(u,d) }  \leq 0,
\end{array}
\end{array}
\end{equation}
{where the term $-\nabla_x h^{-1}\nabla_u h$ corresponds to the \textit{high-voltage} solution sensitivity $\nabla_u \chi(u,d)$ when $x=\chi(u,d)$ \cite{krantz2012implicit}.
%, i.e., $(-\nabla_x h^{-1}\nabla_u h)(\chi(u,d),u,d)=\nabla_u \chi(u,d)$, 
Hence, at $x=\chi(u,d)$, the expression in the constraint \eqref{eq:qpaproxhes} corresponds to a first-order approximation of $g(u^+,d)$.
} 

Since $B(u)\succ 0$, if \eqref{eq:qpaproxhes} is feasible, then it has a unique global minimum, and $T(\cdot)$ is a single-valued function. A large $B(u)$ penalizes deviating from the current operating point, where the first-order approximation is most accurate. %A typical choice for $B(u)$ in sequential quadratic programming \cite{nocedal2006numerical} is an approximation of the Lagrangian Hessian of \eqref{eq:optprob}. Yet, this Hessian requires second-order derivatives of $\chi(u,d)$, which are hard to derive analytically.\comment{remove? introduce as metric?}  % at $u,x,\theta,\lambda$, where $\lambda$ denotes the dual variables, i.e., $B(u,x,\theta,\lambda)=\nabla_{uu}^2 L(u,x,\theta,\lambda)$.  
Convex problems such as \eqref{eq:qpaproxhes} can be solved efficiently and reliably in real-time applications, even when deployed on embedded systems \cite{domahidi2013ecos}. In this paper, we use the operator \eqref{eq:qpaproxhes}, and compare its performance against a projected gradient descent operator based on the soft penalization \eqref{eq:projgrad}, used in previous work \cite{picallo2020seopfrt} in Section~\ref{sec:test}.

\subsection{Dynamic estimation block}\label{sec:estim}

As we only measure $y_t$ in \eqref{eq:meas}, an estimator is needed to reconstruct the estimates $\hat{x}_t$ and $\hat{d}_t$ of the signals $x_t$ and $d_t$. First, we model the evolution of the disturbance $d$ using the following linear stochastic process:
\begin{equation}\label{eq:stochproc}
    d_{t+1} = d_{t} + \mu_{t} + \omega_{t},
\end{equation}
where $\mu_{t}$ is the (known) drift term in \eqref{eq:KFsimple}, and $\omega_{t}$ is assumed to be a bounded zero-mean white-noise with covariance $\Sigma_{\omega,t}$. {In short, \eqref{eq:stochproc} represents the temporal variation of the loads.}

Then, we use an Extended Kalman filter \cite{reif1998ekf} to define the operator $E(\cdot)$ in \eqref{eq:KFsimple} that updates the estimate $\hat{d}_t$ based on $y_t$. When evaluated at an approximate \textit{high-voltage} solution $z_{t-1}$, the operator $E(\cdot)$  can be expressed as
\begin{equation*}
\begin{array}{l}
    % \hat{d}_{t}= & %\hat{d}_{t-1}+\mu_{t-1} + K_{t}\big(y_{t} - \bar{c}(\chi(u_t ,\hat{d}_{t-1}+\mu_{t-1}))\big) \\
    % \hat{d}_{t-1}+\mu_{t-1} + K_{t}\big(y_{t} - \bar{c}(z_{t-1})\big) \\
    % %  \hat{x}_{t} = & \chi(u_t ,\hat{d}_{t})
    %  P_{t} = & (I-K_{t} C_{t}) (P_{t-1} + \Sigma_{\omega,t-1}),
        E(t,\hat{d}_{t-1}+\mu_{t},P_{t-1},u_t,z_{t-1},y_t) = \\
    %\hat{d}_{t-1}+\mu_{t-1} + K_{t}\big(y_{t} - \bar{c}(\chi(u_t ,\hat{d}_{t-1}+\mu_{t-1}))\big) \\
    \Big( \hspace{-0.05cm} \hat{d}_{t-1} \hspace{-0.05cm} + \hspace{-0.05cm} \mu_{t-1} \hspace{-0.05cm} + \hspace{-0.05cm} K_{t}\big(y_{t} \hspace{-0.05cm} - \hspace{-0.05cm} \bar{c}(z_{t-1}) \hspace{-0.05cm} \big), 
    %  \hat{x}_{t} = & \chi(u_t ,\hat{d}_{t})
     (I \hspace{-0.05cm} - \hspace{-0.05cm} K_{t} C_{t}) (P_{t-1} \hspace{-0.05cm} + \hspace{-0.05cm} \Sigma_{\omega,t-1}) \hspace{-0.1cm} \Big),
\end{array}  
% \vspace{-0.1cm}
\end{equation*}
with
\vspace{-0.3cm}
\begin{equation*}
\begin{array}{rl}
C_t = & %\nabla_c c(u_t , \hat{d}_{t-1}+\mu_{t-1}) \\
 \overbrace{\nabla_x \bar{c}(z_{t-1}) (-\nabla_x h^{-1}\nabla_d h)(z_{t-1},u_t ,\hat{d}_{t-1}+\mu_{t-1})}^{=\nabla_d c(u_t , \hat{d}_{t-1}+\mu_{t-1}) \text{ at } z_{t-1}=\chi(u_t , \hat{d}_{t-1}+\mu_{t-1}) }  \\
    K_{t} = & \big( C_{t}^T \Sigma_{\nu,t}^{-1}  C_{t} + (P_{t-1} + \Sigma_{\omega,t-1})^{-1}\big)^{-1}
     C_{t}^T \Sigma_{\nu,t}^{-1},
\end{array}  
\end{equation*}
where $C_t$ extends the sensitivity $\nabla_d c(u_t , \hat{d}_{t-1}+\mu_{t-1})$ of the measurement function $c(u,d)=\bar{c}(\chi(u,d))$ to approximations $z_{t-1}$ that do not necessarily satisfy the power flow equations, i.e., $h(z_{t-1},u_t ,\hat{d}_{t-1}+\mu_{t-1}) \neq 0$. 

\subsection{Online power flow solver block}\label{sec:PFsolver}

The combined OFO \eqref{eq:oper}, estimation \eqref{eq:KFsimple} and power flow solver \eqref{eq:rootfind} is a two-time-scale system: {Newton's method \eqref{eq:rootfind} is executed until convergence before} the optimization and estimation steps \eqref{eq:oper} and \eqref{eq:KFsimple} are executed. However, performing a large number of Newton iterations may jeopardize the real-time feasibility of the controller. 
% on a slow time scale with steps $t$, and \eqref{eq:rootfind} on a faster time scale to compute the power flow solutions. Since executing \eqref{eq:rootfind} until convergence can take any arbitrary number of iterations, these two time scales compromise the real-time application of this scheme, see Remark~\ref{rem:ofort}. 
Instead, we leverage the recently proposed sensitivity-conditioning approach \cite{picallo2021predictivesensitivity} to design an \emph{approximate} online power flow solver in \eqref{eq:rootfindPSint}. This solver is computationally efficient, runs simultaneously with \eqref{eq:oper} and \eqref{eq:KFsimple}, i.e., as $u_{t},\hat{d}_{t}$ transition to $u_{t+1},\hat{d}_{t+1}$, and requires a limited number of iterations per time-step. {In practice, a single iteration suffices to ensure convergence of the approach, as discussed later in Remark~\ref{rem:smallN}. Hence, we first present a single-iteration operator $\pfop{}(\cdot)$, and explain the general $N$-iterations $\pfop{}_N(\cdot)$ later in the analysis Section~\ref{sec:conv}.} 
%As first step to derive the method, 

Consider a first-order approximation of $h(z_{t+1},u_{t+1},\hat{d}_{t+1})$ around $(z_t,u_{t},\hat{d}_{t})$:
% \begin{equation*}\begin{array}{l}
% h_{t+1} \approx \\
% h_t + \nabla_x h_t (z_{t+1} \hspace{-0.05cm} - \hspace{-0.05cm} z_{t}) + \nabla_u h_t (u_{t+1} \hspace{-0.05cm} - \hspace{-0.05cm} u_{t}) + \nabla_d h_t (\hat{d}_{t+1} \hspace{-0.05cm} - \hspace{-0.05cm} \hat{d}_{t}), 
% \end{array}
% \end{equation*}
\begin{equation*}\begin{array}{l}
h(z_{t+1},u_{t+1},\hat{d}_{t+1}) \approx h(z_{t},u_{t},\hat{d}_{t})
+ \nabla h(z_{t},u_{t},\hat{d}_{t}) \left[ \hspace{-0.1cm} \begin{smallmatrix} 
     {z_{t+1}-z_{t}} \\ 
     {u_{t+1}-u_{t}} \\
     {\hat{d}_{t+1}-\hat{d}_{t}} \end{smallmatrix} \hspace{-0.1cm} \right] \hspace{-0.05cm}
\end{array}
\end{equation*}
{The operator $\pfop{}(\cdot)$ in \eqref{eq:rootfindPSint} corresponds to setting this approximation to $0$ to derive an update rule for $z_{t+1}$ that approximates $\chi(u_{t+1},\hat{d}_{t+1})$, i.e.,
%so that $h(z_{t+1},u_{t+1},\hat{d}_{t+1}) \approx 0$:
\begin{equation}\label{eq:taylorstep}\begin{array}{l}
     \pfop{}(z_{t},u_{t+1},u_{t},\hat{d}_{t+1},\hat{d}_t) = \\
    z_t - (\nabla_x h_t)^{-1} \big(h_t + \nabla_u h_t (u_{t+1}-u_{t}) + \nabla_d h_t (\hat{d}_{t+1}-\hat{d}_{t}) \big),
\end{array}
\end{equation}
where $h_{t}=h(z_{t},u_{t},\hat{d}_{t})$, and similarly for $\nabla h_{t}$. The operator $\pfop{}(\cdot)$ corresponds to a Newton iteration \eqref{eq:rootfind} with additional sensitivity-conditioning terms based on the variations $u_{t+1}-u_{t}$ and $\hat{d}_{t+1}-\hat{d}_{t}$, and the sensitivities $\nabla_u \chi(u_d,\hat{d}_{t})= - (\nabla_x h_t)^{-1}\nabla_u h_t$ and $\nabla_d \chi(u_d,\hat{d}_{t})= - (\nabla_x h_t)^{-1}\nabla_d h_t$ \cite{picallo2021predictivesensitivity}. The operator $\pfop{}(\cdot)$ updates $z_t$, while simultaneously compensating for any variations in  $u_t$ and $\hat{d}_t$ due to the OFO and estimation steps \eqref{eq:oper} and \eqref{eq:KFsimple}.} 

A complete step of the combined OFO and estimation with online power flow solvers is summarized in Algorithm~\ref{alg:optest}, where $z_t$ approximates $\hat{x}_t=\chi(u_t,\hat{d}_t)$, and $z_{t}^{\text{pr}}$ the intermediate prior estimation $\chi(u_t,\hat{d}_{t-1}+\mu_t)$.

\begin{algorithm}[H]\label{alg:optest}
\caption{OFO and estimation with online power flow solver (blue block in Fig.~\ref{fig:bldiag}) at time $t$}\label{alg:optest}
\begin{algorithmic}[1]
\STATE \textbf{Measure from grid:} $y_{t},\munderbar{u}_t,\bar{u}_t,u_t$ % (= [u_{t-1}^+]^{\bar{u}_{t}}_{\munderbar{u}_{t}})$
%\STATE $$
\STATE \eqref{eq:rootfindPSint}: $z_{t}^{\text{pr}} \hspace{-0.05cm} = \hspace{-0.05cm} \pfop{}(\hspace{-0.05cm}z_{t-1}, u_t , u_{t-1} ,\hat{d}_{t-1}\hspace{-0.05cm}+\hspace{-0.05cm} \mu_{t-1}, \hat{d}_{t-1}\hspace{-0.05cm})$
\STATE \eqref{eq:KFsimple}: $(\hat{d}_{t},P_t) =E(t,\hat{d}_{t-1}+\mu_{t-1},P_{t-1},u_t,z_{t}^{\text{pr}},y_t)$ %$[\hat{d}_{t},P_t] = E(\hat{d}_{t-1}+\mu_{t},u_t,z_{t}^{\text{pr}},y_t)$ 
%$\hat{d}_{t} = \hat{d}_{t-1}+\mu_{t-1} + \hat{K}_{t}\big(y_{t} - c(\hat{x}_{t-1}^+)\big)$ \label{algstep:est}
\STATE \eqref{eq:rootfindPSint}: $z_{t} = \pfop{}(z_{t}^{\text{pr}}, u_t , u_t ,\hat{d}_{t},\hat{d}_{t-1}+\mu_{t-1})$ \label{algstep:PF2}
\STATE \eqref{eq:oper}: $u_{t}^+ = T(u_t , z_{t}, \hat{\theta}_{t}) $ \label{algstep:opt} 
\STATE \textbf{Output:} $u_{t}^+$ 
\end{algorithmic}
\end{algorithm}

\section{Convergence and Stability Analysis}\label{sec:conv}

In this section, we analyze the convergence and stability of Algorithm~\ref{alg:optest}, including the OFO \eqref{eq:oper}, estimation \eqref{eq:KFsimple}, and online power flow solver \eqref{eq:rootfindPSint}. For this, we need some concepts from nonlinear control \cite[Ch.4]{khalil2002nonlinear}:
\begin{defi}[Comparison functions]
A $\mathcal{K}$-function $\alpha(\cdot)$ is continuous, strictly increasing and satisfies $\alpha(0)=0$. %If $\alpha(\cdot)$ is also unbounded, i.e., $\lim_{r \to \infty} \alpha(r) = \infty$, it is a $\mathcal{K}_{\infty}$-function.
A $\mathcal{L}$-function $\sigma(\cdot)$ is continuous, strictly decreasing and satisfies $\lim_{s \to \infty} \sigma(s) = 0$. 
A $\mathcal{KL}$-function $\beta(\cdot,\cdot)$ satisfies: $\beta(\cdot,s)$ is a $\mathcal{K}$-function for any $s$, and $\beta(r,\cdot)$ is a $\mathcal{L}$-function for any $r$.
    % \item Input-to-State-Stability (ISS) \cite{jiang2001input}: A system $z_{t+1}=f(z_t,v_t)$ with $f(0,0)=0$ is ISS if there exists a $\mathcal{KL}$-function $\beta(\cdot)$, and a $\mathcal{K}$-function $\gamma(\cdot)$ such that
    % \begin{equation}
    %     \norm{z_t}_2 \leq \beta(\abs{z_0}_2,t) + \gamma(\sup)
    % \end{equation}
\end{defi}

%A complete step of this approach can be expressed via the following Algorithm~\ref{alg:optest}, where $z_t$ approximates $\hat{x}_t=\chi(u_t,\hat{d}_t)$, and $z_{t}^{\text{pr}}$ the intermediate prior estimation $\chi(u_t,\hat{d}_{t-1}+\mu_t)$:}
% \noindent{.} %\comment{alg extra with close loop with system}

\subsection{Online power flow convergence}\label{subsec:hconv}

First, we analyze {the convergence of the power flow approximate $z_t$, when using the online power flow solver $\pfop{}(\cdot)$ \eqref{eq:rootfindPSint} in of Algorithm~\ref{alg:optest}. Since the power flow equations $h(\cdot)$ are nonlinear, the approximation error in \eqref{eq:taylorstep} may be arbitrarily large depending on $h_t$, and the variations $u_{t+1}-u_{t}, \hat{d}_{t+1}-\hat{d}_{t}$. Therefore, for the analysis of $\pfop{}(\cdot)$, \r{consider an operator $\pfop{}_N(\cdot)$ that results from splitting $\pfop{}(\cdot)$ in $N$ steps to refine the approximation in \eqref{eq:taylorstep}. More concretely,} define the intermediate values $u_{t,k} =
{u_{t}} + \frac{k}{N}({u_{t+1}} - {u_{t}} )$ for $k \in \{0,...,N\}$, similarly $\hat{d}_{t,k}$, and define the operator $\pfop{}_N(\cdot)$
% \begin{equation}\label{eq:rootfindPSint}
% z_{t+1} = \pfop{}_N(z_{t},{u_{t+1}},{u_{t}},\hat{d}_{t+1},\hat{d}_{t}),
% \end{equation}
consisting of $N$ steps like \eqref{eq:taylorstep}, \r{but each} multiplied by a factor $\frac{1}{N}$:}
% \begin{equation*}\begin{array}{rl}
%      z_{t,k+1} 
%     %  = & \pfop{}(k,N,z_{t,k},{u_{t+1}},{u_{t}},\hat{d}_{t+1},\hat{d}_{t}) \\
%      = & z_{t,k} - \frac{1}{N} (\nabla_x h_{t,k})^{-1} \big( h_{t,k} + \nabla_u h_{t,k} ({u_{t+1}}-{u_{t}}) \\
%     & \hspace{3.3cm} + \nabla_d h_{t,k} (\hat{d}_{t+1}-\hat{d}_{t}) \big),
% \end{array}
% \end{equation*}
% {where $z_{t,k}$ are intermediate values within the operator $\pfop{}_N(\cdot)$, $z_{t}=z_{t,0}$ and $z_{t+1}=z_{t,N}=\pfop{}_{N}(z_{t},u_{t+1},u_{t},\hat{d}_{t+1}, \hat{d}_t)$.}

\begin{algorithm}[H]\label{alg:onlinepf}
\caption{Online power flow solver operator $\pfop{}_N(\cdot)$}\label{alg:onlinepf}
\begin{algorithmic}[1]
\REQUIRE $N,z_{t},u_{t+1},u_{t},\hat{d}_{t+1},\hat{d}_t$
\STATE $z_{t,0}=z_t$
\FOR{$k =0 \to N-1$}
\STATE $h_{t,k} = h(z_{t,k},{u_{t,k}}, {\hat{d}_{t,k}})$, $\nabla h_{t,k} = \nabla h(z_{t,k},{u_{t,k}}, {\hat{d}_{t,k}})$
\STATE $z_{t,k+1} 
     = z_{t,k} \hspace{-0.05cm} - \hspace{-0.05cm} \frac{1}{N} (\nabla_x h_{t,k})^{-1} \hspace{-0.05cm} \big( h_{t,k} \hspace{-0.05cm} + \hspace{-0.05cm} \nabla_u h_{t,k} ({u_{t+1}} \hspace{-0.05cm} - \hspace{-0.05cm} {u_{t}})$ \\ 
    $\hspace{4.5cm} + \nabla_d h_{t,k} (\hat{d}_{t+1} \hspace{-0.05cm} - \hspace{-0.05cm} \hat{d}_{t}) \big)$
% \STATE $\begin{array}{rl}
%      z_{t,k+1} 
%      = & z_{t,k} \hspace{-0.05cm} - \hspace{-0.05cm} \frac{1}{N} (\nabla_x h_{t,k})^{-1} \hspace{-0.05cm} \big( h_{t,k} \hspace{-0.05cm} + \hspace{-0.05cm} \nabla_u h_{t,k} ({u_{t+1}} \hspace{-0.05cm} - \hspace{-0.05cm} {u_{t}}) \\
%     & \hspace{3.cm} + \nabla_d h_{t,k} (\hat{d}_{t+1} \hspace{-0.05cm} - \hspace{-0.05cm} \hat{d}_{t}) \big)
% \end{array}$
\ENDFOR
\RETURN  $z_{t,N}$ 
\end{algorithmic}
\end{algorithm}

% The operator $\pfop{}_N(\cdot)$ corresponds to a $\frac{1}{N}$-step-size Euler-forward integration \cite{atkinson2008introduction} of a continuous-time version of Newton's method \eqref{eq:rootfind} with sensitivity-conditioning terms \cite{picallo2021predictivesensitivity} to compensate the temporal variation of $u_t$ and $\hat{d}_t$, see Appendix~\ref{app:rootfindpredsens}. 
% see Appendix in the extended version \cite{}.

% {The operator $\pfop{}_N(\cdot)$ is based on a continuous-time version of Newton's method \eqref{eq:rootfind} with sensitivity-conditioning \cite{picallo2021predictivesensitivity}, see Appendix~\ref{app:rootfindpredsens} for more details. It is essentially a time discretization of this continuous-time approach with $N$ time steps of size $\frac{1}{N}$ \cite{atkinson2008introduction}. Hence, a $\pfop{}_N(\cdot)$ with a sufficiently large $N$, i.e., a discretization with more and shorter steps, preserves the continuous-time convergence:} % of the sensitivity-conditioning approach: %converges to a neighborhood of the power flow solution:
% {This $N$-step approximation can be interpreted as a time discretization of a continuous-time version of Newton's method \eqref{eq:rootfind} with sensitivity-conditioning \cite{picallo2021predictivesensitivity}, see Appendix~\ref{app:rootfindpredsens} for more details.}  
{This $N$-step approximation can be interpreted as a time discretization of a continuous-time version of Newton's method \eqref{eq:rootfind} with sensitivity-conditioning \cite{picallo2021predictivesensitivity}, see the Appendix for details.
For a sufficiently large $N$, $z_t$ converges to $x_t$:}
    
\begin{prop}\label{prop:hconv} 
Under Assumption~\ref{ass:existsol}, there exists $R_h,L_h>0$ such that if $\normsz{h(z_{0},u_{0},\hat{d}_{0})}_2 \leq R_h$, $\frac{L_h}{N}\leq R_h $, {and $F_N(\cdot)$ is used in Algorithm~\ref{alg:optest} instead of $F(\cdot)$}; then $\nabla_x h(z_{t},u_{t},\hat{d}_{t})$ is invertible for all $t$, and exists a $\mathcal{KL}$-function $\beta(\cdot)$ such that
\begin{equation}\label{eq:hconv}\begin{array}{rl}
    \normsz{h(z_{t},u_{t},\hat{d}_{t})}_2 \leq & \max \big( \beta(\normsz{h(z_{0},u_{0},\hat{d}_{0})}_2,t),\tfrac{L_h}{N} \big), % \\
    %  \underset{t \to \infty}{\to} & 
\end{array}
\end{equation}
and hence $\lim_{t \to \infty} \normsz{h(z_{t},u_{t},\hat{d}_{t})}_2 \leq \tfrac{L_h}{N}$.

% $\normsz{h(z_{t},u_{t},\hat{d}_{t})}_2$ is strictly decreasing on $t$ while $\normsz{h(z_{t},u_{t},\hat{d}_{t})}_2>\frac{L_h}{N}$, and $\lim_{t \to \infty} \normsz{h(z_{t},u_{t},\hat{d}_{t})}_2 \leq \frac{L_h}{N}$.} 
%satisfies $\lim_{t \to \infty} \normsz{h(z_{t},u_{t},\hat{d}_{t})}_2 \leq \frac{L_h}{N}$, and is strictly decreasing on $t$ }
 %there exist $\epsilon \in (0,1)$ and a positive integer $N$, such that $\normsz{h(z_{t+1},u_{t+1},\hat{d}_{t+1})}_2 \leq \epsilon \normsz{h(z_{t},u_{t},\hat{d}_{t})}_2$ if $\normsz{h(z_{t},u_{t},\hat{d}_{t})}>0$, and $\normsz{h(z_{t+1},u_{t+1},\hat{d}_{t+1})}_2 \leq R_h$ otherwise, where $z_{t+1}=\pfop{}_{N}(z_{t}, \allowbreak u_{t+1},u_{t},\hat{d}_{t+1},\hat{d}_t)$.  
\end{prop}
    
\begin{proof}
See the Appendix.
\end{proof}

Proposition~\ref{prop:hconv} establishes that if $z_{0}$ is initialized sufficiently close to the power flow solution $\chi(u_{0},\hat{d}_0)$, and the number of steps $N$ of the online power flow solver $F_N(\cdot)$ is sufficiently large, then Algorithm~\ref{alg:optest} is well-defined, because the Jacobian $\nabla_x h(\cdot)$ used in the operators $T(\cdot),E(\cdot),\pfop{}_N(\cdot)$ remains invertible. Furthermore, since the right-hand-side of \eqref{eq:hconv} decreases as $t$ and $N$ increase, the online power flow solver $\pfop{}_N(\cdot)$ brings the approximation $z_t$ to an arbitrarily small neighborhood of the solution $\hat{x}_t=\chi(u_t,\hat{d}_t)$. \r{This is only possible if the number of internal steps $N$ can be sufficiently large. Hence, Algorithm~\ref{alg:optest}} preserves the convergence of a faster time-scale Newton's method \eqref{eq:rootfind}, independently of the temporal variations of $u_t$ and $\hat{d}_t$, and we can safely use the approximate $z_t$ in Algorithm~\ref{alg:optest}, instead of exactly computing $\chi(u_t ,\hat{d}_{t-1}+\mu_{t-1})$ and $\chi(u_t ,\hat{d}_{t})$ within the estimation \eqref{eq:KFsimple}. 

\begin{rem}\label{rem:smallN}
{The parameter $L_h$ in Proposition~\ref{prop:hconv} is the smallest upper bound on the second-order error $\mathcal{O}\big(\max(\normsz{h_t}_2^2,\normsz{u_{t+1}-u_{t}}_2^2, \allowbreak \normsz{\hat{d}_{t+1}-\hat{d}_{t}}_2^2\big)$ of the approximation in \eqref{eq:taylorstep}, see the Appendix. The constant $L_h$ is smaller for small variations of $u$ and $d$, and when $h(z,u,d) \approx 0$, i.e., $z$ is close to the \textit{high-voltage} solution $\chi(u,d)$. 
% see Appendix in the extended version \cite{}.
Hence, there is a trade-off between the number of iterations $N$, and how fast the disturbance $d$ can change or how aggressively we can control the input $u$ variations through the OFO \eqref{eq:oper}. Given fast control-loop rates, we can reasonably expect the time-variation of $d$ to be small, plus a large $B(u)\succ 0$ in \eqref{eq:qpaproxhes} ensures that any variations of $u$ are also small. In this case, as we will verify in simulation (Section~\ref{sec:test}), $N=1$ suffices to ensure Proposition~\ref{prop:hconv} is applicable and that the power flow solver \eqref{eq:rootfindPSint} is precisely on the same timescale as the OFO \eqref{eq:oper} and the dynamic estimation \eqref{eq:KFsimple}. \r{Hence, in practice we use $\pfop{}(\cdot) = \pfop{}_1(\cdot)$ in Algorithm~\ref{alg:optest}.}}
\end{rem}

\subsection{Convergence of Online Feedback Optimization}

Next, we analyze the input $u_t$ generated by OFO \eqref{eq:oper} in Algorithm~\ref{alg:optest}, when connected with the grid under time-varying conditions as shown in Fig.~\ref{fig:bldiag}. 
Let $\mathcal{U}^s(\theta)$ denote the set of stationary points of \eqref{eq:optprob} for a given $\theta$, i.e., points satisfying the Karush-Kuhn-Tucker first-order optimality conditions \cite{bertsekas1997nonlinear}, and denote the subset of strict local minima as $\mathcal{U}^*(\theta) \subseteq \mathcal{U}^s(\theta)$. %To certify the closed-loop convergence of Algorithm~\ref{alg:optest} 
{We need the following assumption for our convergence result:}
\begin{ass}\label{ass:regsolplus} For all $\theta \in \Theta$ and $u \in \mathcal{U}$: $f(\cdot)$ and $g(\cdot)$ are twice continuously differentiable; there exists $R_\mathcal{U}>0$ such that if $x$ satisfies $\normsz{h(x,u,d)}_2 < R_\mathcal{U}$, then the feasible set $\mathcal{U}^f(\theta) = [\munderbar{u},\bar{u}] \cap \{u^+ | \bar{g}(x) + \nabla_x \bar{g}(x) \nabla_x h(x,u,d)^{-1} \nabla_u h(x,u,d) (u^+ \hspace{-0.05cm} - \hspace{-0.05cm} u) \leq 0\}$ of \eqref{eq:qpaproxhes} is non-empty and satisfies the Linear Independence Constraint Qualification ({LICQ}) condition \cite{bertsekas1997nonlinear} for all $u^+ \in \mathcal{U}^f(\theta)$; and the stationary points in $\mathcal{U}^s(\theta)$ are isolated and Lipschitz continuous on $\theta$.
% \begin{itemize}[leftmargin=*]
%     \item The functions $f(\cdot)$ and $g(\cdot)$ are twice continuously differentiable.
%     \item There exists $R_\mathcal{U}>0$ such that if $x$ satisfies $\normsz{h(x,u,d)}_2 < R_\mathcal{U}$, then the feasible set $\mathcal{U}^f(\theta) = [\munderbar{u},\bar{u}] \cap \{u^+ | \bar{g}(x) + \nabla_x \bar{g}(x) \nabla_x h(x,u,d)^{-1} \nabla_u h(x,u,d) (u^+ \hspace{-0.05cm} - \hspace{-0.05cm} u) \leq 0\}$ of \eqref{eq:qpaproxhes} is non-empty and satisfies the Linear Independence Constraint Qualification ({LICQ}) condition \cite{bertsekas1997nonlinear} for all $u^+ \in \mathcal{U}^f(\theta)$.
%     \item The stationary points in $\mathcal{U}^s(\theta)$ are isolated and Lipschitz continuous on $\theta$. %\comment{revise: Lipschitz cont and strict max and saddle}
%     % \item When $x=\chi(u,d)$, the feasible set of \eqref{eq:qpaproxhes}, denoted as $\mathcal{U}^f = [\munderbar{u},\bar{u}] \cap \{u^+ | \bar{g}(x) + \nabla_x \bar{g}(x) \nabla_x h(x,u,d)^{-1} \nabla_u h(x,u,d) (u^+ \hspace{-0.05cm} - \hspace{-0.05cm} u) \leq 0\}$, is non-empty and satisfies the Linear Independence Constraint Qualification ({LICQ}) condition \cite{bertsekas1997nonlinear} for all $u^+ \in \mathcal{U}^f$.
% \end{itemize}
\end{ass}

Continuously differentiability is usually satisfied by the functions $f(\cdot),g(\cdot)$ used for an AC-OPF \eqref{eq:optprob}. The feasibility and constraint qualification assumptions are typical for these operators \cite{haberle2020non}\cite{torrisi2018projected}, and are usually satisfied when operating close to nominal conditions. Even if \eqref{eq:qpaproxhes} were not feasible, it could be relaxed as in \eqref{eq:projgrad}. %Lipschitz continuity of the stationary points is granted if they have 

In Algorithm~\ref{alg:optest}, the OFO \eqref{eq:oper} is effected by the temporal variation of $\hat{\theta}_t$ and the error between $z_t$ and $\hat{x}_t=\chi(u_t,\hat{d}_t)$. Therefore, we analyse the convergence and stability as a function of the perturbations $\normsz{\hat{\theta}_{t+1}-\hat{\theta}_{t}}_2$ and $\normsz{\hat{x}_t-z_{t}}_2$:

\begin{thm}[Local convergence and stability]\label{thm:opconvloc}%\comment{split off?}
Under Assumptions~\ref{ass:existsol} and \ref{ass:regsolplus}, there exists $\eta>0$ such that if $B(u) \succeq \eta I_d$, then: 
%and is continuous, then %for all $t>0,u_0 \in \mathcal{U}$ 
%\comment{long thm, remove part?} %\comment{add res using prop 1, to show that oper well-def}
\begin{enumerate}[leftmargin=*]
    \item\label{thmit:1} For a constant $\theta$ and using the exact solution $x=\chi(u,d)$ in \eqref{eq:oper}, a stationary point $u^s(\theta) \in \mathcal{U}^s(\theta)$ of \eqref{eq:optprob} is a strict local minimum if and only if it is an asymptotically stable equilibrium of \eqref{eq:oper}. If the stationary point $u^s(\theta)$ is not a strict local minimum, then it is unstable.
    \item\label{thmit:2} Consider a strict local minimum $u^*(\hat{\theta}_t) \in \mathcal{U}^*(\hat{\theta}_t)$. There exists $R_h,L_h,R_u,R_\theta,R_x>0$, such that if $\tfrac{L_h}{N}<R_h$, $\normsz{h(z_{0},u_0,\hat{d}_0)} \leq R_h$,
    %$B(u) \succ \eta I_d$ for $T(\cdot)$ in Algorithm~\ref{alg:optest}; 
    $\normsz{u_0-u^*(\theta_0)}_2 \leq R_u$, and $\frac{1}{R_\theta} \normsz{\hat{\theta}_{t+1}-\hat{\theta}_t}_2 + \frac{1}{R_x} \normsz{\hat{x}_{t}-z_t}_2 \leq 1$ for all $t$; then there exists $\xi \in (0,1)$ and $L>0$, such that the set-points $u_t$ generated by Algorithm~\ref{alg:optest} satisfy
\end{enumerate}
    \begin{equation*}
%     \begin{array}{rl}
%          & \normsz{u_{t+1}-u^*(\theta_{t+1})}_2 \\
%         \leq & \xi \normsz{u_{t}-u^*(\theta_{t})}_2 + L(\normsz{\hat{\theta}_{t+1}-\hat{\theta}_t}_2+\normsz{\hat{x}_{t}-z_t}_2)  \\
%         % \leq & \xi^{t+1} \normsz{u_{0}-u^*(\theta_{0})}_2 
% %        + L\frac{1}{1-\eta}\sup_{k\leq t} (\normsz{\hat{\theta}_{k+1}-\hat{\theta}_k}_2+\normsz{\hat{x}_{k}-z_k}_2)   \\
%     \underset{t \to \infty}{\to} & \frac{L}{1-\xi}\sup_{t} (\normsz{\hat{\theta}_{t+1}-\hat{\theta}_t}_2+\normsz{\hat{x}_{t}-z_t}_2) 
%     \end{array}
\begin{array}{rl}
     \normsz{u_{t}-u^*(\hat{\theta}_{t})}_2 \leq & \xi^{t} \normsz{u_{0}-u^*(\hat{\theta}_0)}_2 \\
     & + \frac{L}{1-\xi} \sup_{k < t} \big( \normsz{\hat{\theta}_{k+1} \hspace{-0.05cm} - \hspace{-0.05cm} \hat{\theta}_{k}}_2 + \normsz{\hat{x}_k \hspace{-0.05cm} - \hspace{-0.05cm} z_k}_2 \big).
\end{array}
    \end{equation*}
\end{thm}

Theorem~\ref{thm:opconvloc} establishes that if the initial set-points $u_0$ are close to a strict local minimum $u^*(\hat{\theta}_0)$, and the values $\normsz{\hat{\theta}_{t+1}-\hat{\theta}_t}_2$ and $\normsz{\hat{x}_{t}-z_t}_2$ are sufficiently small, then the set-points $u_t$ produced by Algorithm~\ref{alg:optest} converge to a neighborhood of the time-varying solution $u^*(\hat{\theta}_t)$ with finite tracking error, i.e., $\lim_{t \to \infty} \normsz{u_{t}-u^*(\hat{\theta}_{t})}_2 \leq \frac{L}{1-\xi} \sup_{t} \big( \normsz{\hat{\theta}_{t+1} \hspace{-0.05cm} - \hspace{-0.05cm} \hat{\theta}_{t}}_2 + \normsz{\hat{x}_t \hspace{-0.05cm} - \hspace{-0.05cm} z_t}_2 \big)$. Since stationary points that are not strict local minima are unstable, even if $u_{t}$ becomes trapped in one, i.e., $u_t \in \mathcal{U}^s(\hat{\theta}_t) \setminus \mathcal{U}^*(\hat{\theta}_t)$, a perturbation will free it. Further, if all stationary points are strict local minima, e.g., if there is a single global minimum, Theorem~\ref{thm:opconvloc} can be extended and we can show that convergence and stability are preserved without any requirements on the starting point $u_0$ or the perturbations:

\begin{thm}[Global convergence and stability]\label{thm:opconvglb}
    Assume the stationary points are all strict local minima, i.e., $\mathcal{U}^*(\theta)=\mathcal{U}^s(\theta)$. Under Assumptions~\ref{ass:existsol} and \ref{ass:regsolplus}, there exists $\eta, R_h,L_h>0$ such that if $B(u) \succeq \eta I_d$ and is continuous, $\tfrac{L_h}{N}<R_h$ and $\normsz{h(z_{0},u_0,\hat{d}_0)} \leq R_h$; then for all $t>0,u_0 \in \mathcal{U}$ there exists a $\mathcal{KL}$-function ${\beta}(\cdot)$ and a $\mathcal{K}$-function ${\gamma}(\cdot)$, such that the set-points $u_t$ generated by Algorithm~\ref{alg:optest} satisfy
\begin{equation*}\begin{array}{rl}
     \distset{u_t}_{\mathcal{U}^*({\hat{\theta}}_t)} \leq &  {\beta}(\distset{u_0}_{\mathcal{U}^*(\hat{\theta}_0)},t) \\
     & + {\gamma} \big(\underset{k < t}{\sup} \normsz{\hat{\theta}_{k+1} \hspace{-0.05cm} - \hspace{-0.05cm} \hat{\theta}_{k}}_2 +  \normsz{\hat{x}_k-z_k}_2 \big),
\end{array}
\end{equation*}
where $\distset{u_t}_{\mathcal{U}^*({\hat{\theta}}_t)}=\min_{u^* \in \mathcal{U}^*({\hat{\theta}}_t)} \normsz{u^*-u_t}$ denotes the distance of $u_t$ to the set $\mathcal{U}^*({\hat{\theta}}_t)$.
\end{thm}

\begin{proof}
See the Appendix for both Theorems.
\end{proof}

% {For an exact power flow solution $z=\chi(u,d)$,} the dynamic estimation \eqref{eq:KF} is exponentially bounded in mean square, i.e., the expectation of the norm $\normsz{\hat{d}_{t}-d_{t}}_2$ converges to a bounded set under some mild technical conditions \cite{reif1998ekf}. Among others, this holds for small enough $\Sigma_{\omega,t},\Sigma_{\nu,t}$, an initialization $\hat{d}_{0}$ close to $d_{0}$, and enough independent measurements in $\bar{c}(x)$. Assuming that this property holds given the diminishing error $\normsz{\chi(u_t,d_t)-z_t} \to 0$, the term...  %Then, since $u$ belongs to a compact set, convergence is preserved independently of $u$.

\section{Test Case}\label{sec:test}
In this section, we simulate and validate the proposed combined OFO and estimation approach in Algorithm~\ref{alg:optest} on a benchmark test feeder. We compare it against a projected gradient descent operator based on \eqref{eq:projgrad} used in \cite{picallo2020seopfrt} during a 1-hour simulation with 1-second sampling period. % i.e., 1 second from  step $t$ to $t+1$.

\begin{figure}[t]
\centering
\includegraphics[width=8cm,height=5.5cm]{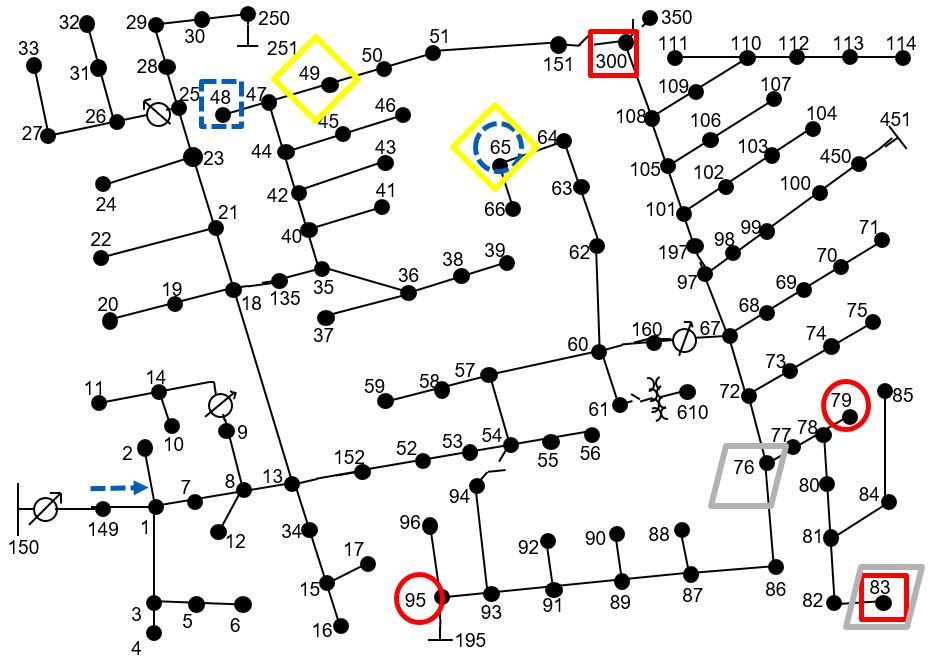}    % The printed column width is 8.4 cm.
\caption{IEEE 123-bus test feeder \cite{kersting1991radial}. \textbf{Measurements:} red circle~=~voltage phasor, red square~=~voltage magnitude, blue dashed circle~=~current phasor, blue dashed square~=~current magnitude, blue dashed arrow~=~line current phasor. \textbf{Distributed generation:} yellow diamond~=~solar, grey parallelogram~=~wind.} 
\label{fig:123bus}
%\vspace{-0.6cm}
\end{figure}

\subsection{Simulation settings}\label{subsec:simDef}

\begin{itemize}[leftmargin=*]
\item System and measurements $y$: We use the 3-phase, unbalanced IEEE 123-bus test feeder \cite{kersting1991radial} (see Fig.~\ref{fig:123bus}). As in \cite{picallo2020seopfrt}, actual measurements are placed at different locations in the grid. 
%Apart from them, distribution grids state estimation typically incorporate so-called pseudo-measurements \cite{schenato2014bayesian, picallo2017twostepSE}, based on load predictions or profiles, 
To complement these measurements and achieve numerical observability \cite{baldwin1993power}, pseudo-measurements are built by averaging the loads profiles in $d$. %Actual measurements are assigned a relatively low uncertainty value (1\% standard deviation) in the corresponding covariance terms in $\Sigma_\nu$, while pseudo-measurement a relatively large value (50\% standard deviation) due to their higher uncertainty.

\item Load profiles for $d$: We use $1$-second resolution data of the ECO data set \cite{ECOdata}, aggregate households and scale them to the base loads of the 123-bus feeder. Hence, we do not base our simulation on a stochastic process like \eqref{eq:stochproc}, but rather use real-world data to make our simulation more realistic. We use $\mu_t=0$ to test the approach in more adverse conditions, instead of estimating $\mu_t$ {through, e.g., time series analysis}.

\item Controllable inputs $u$: As in \cite{picallo2020seopfrt}, solar energy and wind energy is introduced in the three phases the buses indicated in Figure~\ref{fig:123bus}. Their profiles are generated based on a $1$-minute solar irradiation profile and a $2$-minute wind speed profiles from \cite{solarprofile, windprofile}. Generation is assumed constant between samples. We use these profiles to set the time-varying upper and lower limits of the feasible set, $\bar{u}_{t}$ and $\munderbar{u}_{t}$ respectively, and set the lower limit of active generation to $\munderbar{u}_{t}=0$. 

\item Functions $f(u)$ and $g(u,d)$ in the AC-OPF \eqref{eq:optprob}: We use a quadratic cost that penalises deviating from a reference: $f(u)=\frac{1}{2}\normsz{u-u_\text{ref}}_2^2$. As reference $u_\text{ref}$ we use $1$p.u. for the voltage magnitude at the slack bus, and the maximum installed power as reference for the controllable active generation, to promote using as much renewable energy as possible. We consider the voltage limits $\abs{U}_i \in [0.94 \text{p.u.},1.06 \text{p.u.}]$ for all nodes $i$ as in \cite{hauswirth2017onlinePF,picallo2020seopfrt}, and thus use $g(u,d)=\bar{g}(\chi(u,d)) = \left[\begin{smallmatrix} I_d & 0 \\ -I_d & 0 \end{smallmatrix}\right] \chi(u,d) + \left[\begin{smallmatrix} -1.06 \\ 0.94 \end{smallmatrix}\right]  \leq 0$, where $I_d$ is the identity matrix. Note that $f(\cdot)$ is strongly convex and twice continuously differentiable, as requested in Assumption~\ref{ass:regsolplus}. %Also, the time-varying element $\bar{u}$ in $f(u)$ is the same as in the set $[\munderbar{u},\bar{u}]$, hence Prop.~\ref{prop:opconvPS} still holds. 
Furthermore, while $x$ remains close to the \textit{high-voltage} solution $\chi(u,d)$, $\nabla_x h(\cdot)$ is invertible, and thus $g(u,d)$ is also continuously differentiable. 

\item Matrix $B(u)$ in \eqref{eq:qpaproxhes}: We use $B(u)=\eta I_d$, with $\eta>0$ large enough to ensure that the corresponding condition in Theorems~\ref{thm:opconvloc} and \ref{thm:opconvglb} is satisfied. Experimentally, we have observed that $\eta=5$ suffices to preserve convergence.

% \item Steps $N$ in \eqref{eq:rootfindPSint}: Experimentally, we have observed that a single step $N=1$ suffices to preserve the convergence in Proposition~\ref{prop:hconv}, see Remark~\ref{rem:smallN}. 
\end{itemize}

\begin{figure}
\centering
\includegraphics[width=8.8cm,height = 4.5cm]{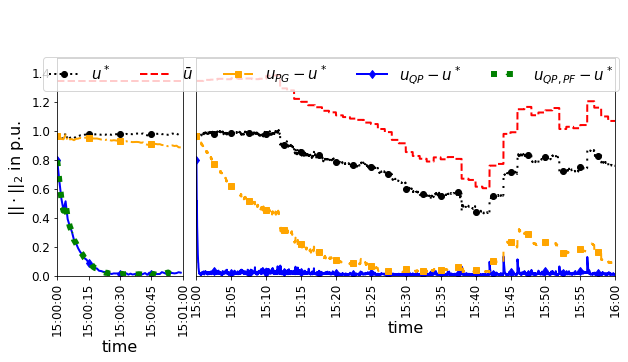}    % The printed column width is 8.4 cm.
\caption{Euclidean norm of the AC-OPF solution $u^*$, the maximum power available $\bar{u}$, and the difference between $u^*$ and the set-points $u_{QP}$ produced by Algorithm~\ref{alg:optest} using the quadratic operator \eqref{eq:qpaproxhes}, or $u_{PG}$ produced the projected gradient descent based on \eqref{eq:projgrad}. On the right, in a zoom-in during the first minute, {we also show this difference when using Algorithm~\ref{alg:optest}, but solving the power flow until convergence $u_{QP,PF}$.}} 
\label{fig:uerror}
\end{figure}

\begin{figure}
\centering
\includegraphics[width=8.8cm,height = 3cm]{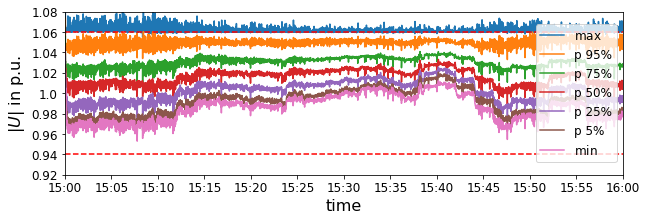}    % The printed column width is 8.4 cm.
\caption{Percentiles of the true voltage magnitude in $x$ during the simulation.}
\label{fig:Vtime}
\end{figure}

\subsection{Results}
We monitor the optimization error $\normsz{u_{t}^*-u_{t}^+}_2 $ over the simulation time, where $u_{t}^*=u^*(\theta_{t})$ is computed by solving \eqref{eq:optprob} given $\theta_t$, i.e., the power limits $\munderbar{u}_t,\bar{u}_{t}$, and the actual disturbance $d_{t}$. We consider the input $u_t^+$ generated by Algorithm~\ref{alg:optest} using both the quadratic operator $T(\cdot)$ in \eqref{eq:qpaproxhes}, and a projected gradient descent operator based on \eqref{eq:projgrad} as in \cite{picallo2020seopfrt}, denoted respectively as $u_{QP}$ and $u_{PG}$.

In Figure~\ref{fig:uerror} we observe how the power available $\bar{u}$ and thus the AC-OPF solution $u^*$ vary over time. We observe that with both operators the optimization errors decrease to close to $0$, but not exactly $0$, due to the effects of the time-varying conditions in $\theta$ and the estimation error $\normsz{d_t-\hat{d}_t}_2$, see Theorems~\ref{thm:opconvloc} and \ref{thm:opconvglb}. Yet, while the inputs $u_{QP}$ of quadratic operator \eqref{eq:qpaproxhes} need only around $30$ seconds (hence time steps) to converge, the projected gradient descent $u_{PG}$ requires about $30$ minutes. A plausible explanation for this slow convergence is that the large $\rho \gg 0$ in \eqref{eq:projgrad} (required to avoid constraint violations) increases the Lipschitz constant of the cost function $\phi(\cdot)$ in \eqref{eq:projgrad}, and hence limits the projected gradient descent step size. On the other hand, the quadratic programming operator \eqref{eq:qpaproxhes} can handle constraints directly and efficiently. This can be observed again from 15:45 to 16:00, when suddenly more energy $\bar{u}$ becomes available. Then, the total amount of energy used by $u^*$ increases, the quadratic operator $u_{QP}$ adapts quickly to the new $u^*$, while the projected gradient descent $u_{PG}$ requires more time. {We also observe in the zoom-in in Figure~\ref{fig:uerror}, that solving the power flow until convergence does not improve the performance significantly, which verifies that our single-iteration online power flow solver \eqref{eq:rootfindPSint} is sufficient to ensure the convergence of Algorithm~\ref{alg:optest}, see Remark~\ref{rem:smallN}.}

Additionally, in Figure~\ref{fig:Vtime} we observe how the voltage magnitudes $|U|$ remain within limits for most nodes. There are some violations caused again by the time-varying conditions, and the estimation error, see Theorems~\ref{thm:opconvloc} and \ref{thm:opconvglb}.

\section{Conclusions and Outlook}\label{sec:conc}

In this paper, we proposed an algorithm for real-time grid operation that combines Online Feedback Optimization (OFO) and dynamic estimation with an online power flow solver. First, we have shown that our online power flow solver is stable and converges despite having control inputs and disturbances that change simultaneously. Then, we have proved that a quadratic-programming-based OFO converges to a neighborhood of the AC-OPF solutions, {with a quantifiable small tracking error depending on the size of exogenous disturbances, and the power flow approximation}. Finally, we have observed in a simulated test case how our approach succeeds in tracking the solution of the AC-OPF, and adapts very quickly to time-varying conditions. An interesting future direction could be to account for estimation uncertainty to take more robust decisions and lower the occurrence of voltage violations.

\bibliographystyle{IEEEtran}
\bibliography{IEEEabrv,ifacconf}

%\iffalse

\appendices

\section{Sensitivity-conditioning for Power Flow}\label{app:rootfindpredsens}

{The operator $\pfop{}_N(\cdot)$ in Algorithm~\eqref{alg:onlinepf} corresponds to a $\frac{1}{N}$-step-size Euler-forward integration \cite{atkinson2008introduction} of a continuous-time version of Newton's method \eqref{eq:rootfind} with sensitivity-conditioning terms \cite{picallo2021predictivesensitivity} to compensate the temporal variation of $u_t$ and $\hat{d}_t$.} 

To illustrate this, first we represent the continuous-time version of Newton's method in \eqref{eq:rootfind} as
\begin{equation}\label{eq:newtoncont}
    \frac{dx(\tau)}{d\tau} = -\nabla_x h(x(\tau),u,d)^{-1} h(x(\tau),u,d),
\end{equation}
where $u$ and $d$ remain constant in the faster time scale $\tau$, as in the discrete-time iterations in \eqref{eq:rootfind}. The \textit{high-voltage} solution $\chi(u,d)$ is an equilibrium of \eqref{eq:newtoncont}, since $h(\chi(u,d),u,d)=0$. Furthermore, since $\nabla_x h(\chi(u,d),u,d)$ is invertible under Assumption~\ref{ass:existsol}, and under \eqref{eq:newtoncont} the norm of the power flow equations is exponentially decreasing in a neighborhood of $\chi(u,d)$, i.e., $\frac{d\normsz{h(x(\tau),u,d)}_2^2}{d\tau} = -2\normsz{h(x(\tau),u,d)}_2^2$, then $\chi(u,d)$ is a locally exponentially stable equilibrium of \eqref{eq:newtoncont}. After \eqref{eq:newtoncont} converges to $\chi(u,d)$, the updates of $u,d$ can take place on a slower time scale via \eqref{eq:KFsimple} and \eqref{eq:oper}. 

Consider $\delta \in [0,1]$ and a continuous transition from the initial input $u_{t}$ to the final one $u_{t+1}$: $u(\delta)= {u_{t}} + \delta({u_{t+1}}-{u_{t}})$, and similarly $\hat{d}(\delta)$. Then the slower continuous-time dynamics of $u,d$ representing their temporal variation are
\begin{equation}\label{eq:slowdyn}\begin{array}{rl}
    \frac{du(\delta)}{d\delta} = ({u_{t+1}}-{u_{t}}), \; \frac{d\hat{d}(\delta)}{d\delta} = ({\hat{d}_{t+1}}-{\hat{d}_{t}}).
\end{array}
\end{equation}

The sensitivity-conditioning approach \cite{picallo2021predictivesensitivity} allows to merge the fast dynamics of $x(\tau)$ in \eqref{eq:newtoncont} with the slow ones of $u(\delta),d(\delta)$ in \eqref{eq:slowdyn}, while preserving the stability of \eqref{eq:newtoncont}. For this, feedforward terms based on $\frac{du(\delta)}{d\delta},\frac{d\hat{d}(\delta)}{d\delta}$ are added to \eqref{eq:newtoncont}:

\begin{equation}\label{eq:newtoncontPS}
    \frac{dx{(\delta)}}{d\delta} = -\nabla_x h{(\delta)}^{-1} \big( h{(\delta)} + \nabla_u h{(\delta)} \tfrac{du{(\delta)}}{d\delta} + \nabla_d h{(\delta)} \tfrac{d\hat{d}{(\delta)}}{d\delta} \big),
\end{equation}
where in an abuse of notation we use $h(\delta)=h(x(\delta),u(\delta),\hat{d}(\delta))$, and similarly $\nabla h{(\delta)}$. Under \eqref{eq:newtoncontPS} and \eqref{eq:slowdyn} we have $\frac{d\normsz{h{(\delta)}}_2^2}{d\delta}=-2\normsz{h{(\delta)}}_2^2$. Hence, $\normsz{h{(\delta)}}_2 = \normsz{h{(0)}}_2 e^{-\delta}$, and $x(\delta)$ converges exponentially to $\chi(u(\delta),\hat{d}(\delta))$ despite the time-varying $u(\delta),\hat{d}(\delta)$. 

\section{Proof of Prop.~\ref{prop:hconv}}\label{app:proofhconv}

Consider the first-order approximation of $h_{t,k+1}=h(z_{t,k+1},u_{t,k+1},d_{t,k+1})$ around $(z_{t,k},u_{t,k},\hat{d}_{t,k})$:
\begin{equation*}\begin{array}{rl}
h_{t,k+1} = & h_{t,k}
+ \nabla h_{t,k} \left[ \hspace{-0.1cm} \begin{smallmatrix} 
     {z_{t,k+1}-z_{t,k}} \\ 
     {u_{t,k+1}-u_{t,k}} \\
     {\hat{d}_{t,k+1}-\hat{d}_{t,k}} \end{smallmatrix} \hspace{-0.1cm} \right] + \mathcal{O} \hspace{-0.1cm} \left( \hspace{-0.1cm} \left[ \hspace{-0.1cm} \begin{smallmatrix} 
     \normsz{z_{t,k+1}-z_{t,k}}_2^2 \\ 
     \normsz{u_{t,k+1}-u_{t,k}}_2^2 \\
     \normsz{\hat{d}_{t,k+1}-\hat{d}_{t,k}}_2^2 \end{smallmatrix} \hspace{-0.1cm} \right] \hspace{-0.1cm} \right) \\
    %  \stackrel{\eqref{eq:rootfindPSint}}{=} 
     = & (1-\frac{1}{N})h_{t,k} + \mathcal{O}(\frac{1}{N^2}),
\end{array}
\end{equation*}
where the last equality holds by Algorithm~\ref{alg:onlinepf}. As $\nabla_x h (\chi(u,d),u,d)$ is invertible under Assumption~\ref{ass:existsol}, there exists $R_h>0$ such that $\nabla_x h(x,u,d)$ is invertible if $\normsz{h(x,u,d)}_2 \leq R_h$. Since $h(\cdot)$ is smooth, its second-order derivatives are bounded in the compact sets $\mathcal{U},\mathcal{D}$. Hence, if $\normsz{h_{t,k}}_2 \leq R_h$, the second-order error is bounded by some $L_h>0$, i.e., $\mathcal{O}(\frac{1}{N^2}) \leq \frac{L_h}{N^2}$, and thus applying \eqref{eq:1stapprox} we get
\begin{equation}\label{eq:1stapprox}
    \normsz{h_{t,k+1}}_2 \leq (1-\tfrac{1}{N})\normsz{h_{t,k}}_2 + \tfrac{L_h}{N^2}.
\end{equation}

First, note that if $\normsz{h_{t_0,k_0}}_2 \leq \frac{L_h}{N} \leq R_h$ for some $t_0,k_0$, then it remains below this bound, since $$
\normsz{h_{t_0,k_0+1}}_2 \leq (1-\tfrac{1}{N})\normsz{h_{t_0,k_0}}_2 + \tfrac{L_h}{N^2} \leq (1-\tfrac{1}{N})\tfrac{L_h}{N} + \tfrac{L_h}{N^2} = \tfrac{L_h}{N}.
$$
On the other hand, while $\frac{L_h}{N} \leq \normsz{h_{t}}_2 \leq R_h$ we have
$$ \normsz{h_{t+1}}_2 - \normsz{h_{t}}_2  \leq -\tfrac{1}{N} \sum_{j=0}^N \overbrace{\normsz{h_{t,j}}_2-\tfrac{L_h}{N}}^{>0}<0, $$
% $$ \normsz{h_{t}}_2 - \normsz{h_{t-1}}_2  \leq -\tfrac{1}{N} \sum_{i=0}^{t-1}\sum_{j=0}^N \overbrace{\normsz{h_{i,j}}_2-\tfrac{L_h}{N}}^{>0}. $$
so $\normsz{h_{t}}_2$ is strictly decreasing on $t$, and thus $\lim_{t \to \infty} \normsz{h_{t}}_2 \leq \tfrac{L_h}{N}$, because $\normsz{h_{t}}_2$ is lower bounded by $0$. As a result, if $\normsz{h(z_{0},u_{0},\hat{d}_{0})}_2 \leq R_h$ and $\frac{L_h}{N}\leq R_h $, then $\nabla_x h(z_{t},u_{t},\hat{d}_{t})$ is invertible for all $t$.

\section{Proof of Thm.~\ref{thm:opconvloc} and \ref{thm:opconvglb}}\label{app:proofopconv}

\subsection{Thm.~\ref{thm:opconvloc}: Local convergence and stability} Consider the following continuous $l_1$ merit function \cite[Ch.~15]{nocedal2006numerical} defined on $\mathcal{U} \times \Theta$, with parameter $\zeta>0$:
\begin{equation*}\begin{array}{l}
V(u,\theta) = \\
f(u) + \zeta \big( \sum_i [g_i(u,d)]_0^\infty  +\sum_i [u_i-\bar{u}_i]_0^\infty + [\munderbar{u}_i-u_i]_0^\infty \big)
\end{array}
\end{equation*} 
Since $f(\cdot)$ and $g(\cdot)$ are twice continuously differentiable over the compact sets $\mathcal{U}$, $\mathcal{D}$, the derivatives $\nabla_u f(\cdot), \allowbreak \nabla_u g(\cdot)$ are Lipschitz continuous with parameters $L_f,L_g$.
Define $\bar{\zeta}$ as the supremum over $u\in \mathcal{U}, \theta \in \Theta$ of all optimum dual variables of \eqref{eq:qpaproxhes}.  
% (I.e., $\bar{\zeta}= \sup_{\theta \in \Theta, u \in \mathcal{U}} \max\big( \max_i \lambda_i^*(u,\chi(u,d),\theta), \allowbreak \max_i \munderbar{\lambda}_i^*(u,\chi(u,d),\theta), \allowbreak \max_i \bar{\lambda}_i^*(u,\chi(u,d),\theta) \big)$, where $\lambda_i^*(\cdot),\munderbar{\lambda}_i^*(\cdot),\bar{\lambda}_i^*(\cdot) \geq 0$ denote the optimal dual variable associated with $\bar{g}_i(x) + \nabla_x \bar{g}_i(x)^T \nabla_u \chi(u,d)(u^+ \hspace{-0.1cm} - \hspace{-0.05cm} u) \leq 0, \munderbar{u}_i \leq {u}_i, u_i \leq \bar{u}_i$ respectively.) 
Under Assumptions \ref{ass:existsol},\ref{ass:regsolplus}, the operator $T(\cdot)$ in \eqref{eq:qpaproxhes} and these optimal dual variables are continuous in $u \in \mathcal{U},\theta \in {\Theta}$ \cite{jittorntrum1984solution,haberle2020non}. Since $\mathcal{U},\Theta$ are compact, $\bar{\zeta}$ is finite. Then, for $\zeta>\bar{\zeta}$, $u \in \mathcal{U}$, $B(u) \succeq \eta I_d$, and $\nu=\eta-L_{f}+L_{g}$, we have:

\begin{equation}\label{eq:descdir}\begin{array}{l}
         V(\overbrace{T(u,\chi(u,d),\theta)}^{u^+},\theta)-V(u,\theta)  \\
         \leq -(\lambda_{\min}(B(u))-L_{f}-L_{g})\normsz{T(u,\chi(u,d),\theta)-u}_2^2 \\
         \leq -\nu \normsz{T(u,\chi(u,d),\theta)-u}_2^2, % \\
        %  \leq -\nu \alpha_1(\normsz{u-u^*(\theta)}_2)
    \end{array}
\end{equation}
see, e.g., \cite[Lemma~18.2]{nocedal2006numerical}, \cite[Lemma~10.4.1]{bazaraa2013nonlinear}, \cite{torrisi2018projected,haberle2020non}.

The stationary points of \eqref{eq:optprob} are the only fixed points of \eqref{eq:qpaproxhes}, i.e., $u=T(u,\chi(u,d),\theta)$ for $u \in \mathcal{U}^s(\theta)$, and $\normsz{T(u,\chi(u,d),\theta)-u}_2^2 \allowbreak >0$ for $u \notin \mathcal{U}^s(\theta)$. Hence, since $T(u,\chi(u,d),\theta)$ is continuous for $u \in \mathcal{U}$ and $\theta \in \Theta$, and positive definite with respect to $\mathcal{U}^s(\theta)$, there exists $\mathcal{K}$-functions $\alpha_{1,\theta}(\cdot)$ and $\alpha_{2,\theta}(\cdot)$ such that $\alpha_{1,\theta}(\distset{u}_{\mathcal{U}^s(\theta)}) \leq \normsz{T(u,\chi(u,d),\theta)-u}_2^2 \leq \alpha_{2,\theta}(\distset{u}_{\mathcal{U}^s(\theta)})$. These $\mathcal{K}$-functions are defined on $[0,r(\theta)]$, where $r(\theta) = \max_{u \in \mathcal{U}} \distset{u}_{\mathcal{U}^s(\theta)}>0$ (see the proof of \cite[Lemma~4.3]{khalil2002nonlinear} for a way to construct them). 

Consider a strict local minimum $u^*(\theta) \in \mathcal{U}^*(\theta)$. For $u$ in a neighborhood of $u^*(\theta)$ we have $\alpha_{1,\theta}(\normsz{u-u^*(\theta)}_2) \leq \normsz{T(u,\chi(u,d),\theta)-u}_2^2 \leq \alpha_{2,\theta}(\normsz{u-u^*(\theta)}_2)$. Define $\tilde{V}(u,\theta)=V(u,\theta)-V(u^*(\theta),\theta)$. Since $u^*(\theta)$ is a strict local minimum, $\tilde{V}(u,\theta)$ is positive definite in a neighborhood of $u^*(\theta)$, and by \eqref{eq:descdir} we have $\tilde{V}(u^+,\theta)-\tilde{V}(u,\theta) \leq -\nu\alpha_{1,\theta}(\normsz{u-u^*(\theta)}_2)$. Hence, $u^*(\theta)$ is an asymptotically stable equilibrium of \eqref{eq:oper} \cite{jiang2004nonlinear}. The other direction of the implication, i.e., that an asymptotically stable equilibrium is a strict local minimum, follows from \cite[Thm.~3]{haberle2020non}.

Consider a stationary point $u^s(\theta)$ that is not a strict local minimum, i.e., $u^s(\theta) \in \mathcal{U}^s(\theta)\setminus \mathcal{U}^*(\theta)$. Since stationary points are isolated by assumption, $u^s(\theta)$ cannot be a local minimum, hence there exists a feasible $u \in \mathcal{U}^f(\theta)$ in the neighborhood of $u^s(\theta)$ such that $f(u) < f(u^s(\theta))$, and thus $V(u,\theta) < V(u^s(\theta),\theta)$. Since $V(u,\theta)$ is strictly decreasing by \eqref{eq:descdir}, then the trajectory moves away from $u^s(\theta)$.

\ref{thmit:2})
First we show that $T(\cdot)$ in \eqref{eq:qpaproxhes} is continuous and locally Lipschitz on $x$ when applying Algorithm~\ref{alg:optest}: Under Assumptions~\ref{ass:existsol} and \ref{ass:regsolplus}, there exists $R_h>0$ such that $\nabla_x h(x,u,d)$ is invertible if $\normsz{h(x,u,d)}_2 < R_h$, and $R_\mathcal{U}>0$ such that LICQ holds while $\normsz{h(x,u,d)}_2 < R_\mathcal{U}$. Define the minimum $R=\min(R_h,R_\mathcal{U})$. By Prop.~\ref{prop:hconv}, if $\tfrac{L_h}{N}\leq R$ and $\normsz{h(z_0,u_0,\hat{d}_0)}_2 < R$, then $\normsz{h(z_{t}^{\text{pr}},u_t ,\hat{d}_{t-1}+\mu_{t-1})}_2 < R$ and $\normsz{h(z_t,u_t ,\hat{d}_t)}_2 < R$ for all $t$, and invertibility of $\nabla_x h (\cdot)$ and LICQ of $\mathcal{U}^f(\hat{\theta}_t)$ are preserved for all $t$. Then, $T(\cdot)$ is continuous and locally Lipschitz on $x$ \cite{jittorntrum1984solution}, i.e., there exists $R_T,L_T>0$ so that $\normsz{T(u_t ,z_t,\hat{{\theta}}_{t})-T(u_t ,\hat{x}_{t},\hat{\theta}_{t})}_2 \leq L_{T} \normsz{\hat{x}_{t} - z_t}_2$ when $\normsz{h(z_t,u_t ,\hat{d}_t)}_2 < R_T$. Hence, we need $\normsz{h(z_0,u_0,\hat{d}_0)}_2\leq \min(R_T,R)$ and $\tfrac{L_h}{N} \leq \min(R_T,R)$.

For a constant $\theta$ and using the exact power flow solution $x=\chi(u,d)$, \eqref{eq:oper} converges linearly in a neighborhood of the strict local minimum $u^*(\theta)$ \cite[Thm.~4.1]{torrisi2018projected}, i.e., there exists $\xi \in(0,1)$ such that $\normsz{T(u,\chi(u,d),\theta)-u^*({\theta})}_2 \leq \xi\normsz{u-u^*({\theta})}_2$. Finally, since $u^*(\hat{\theta})$ is Lipschitz continuous there exists $L_u>0$ such that %Algorithm~\ref{alg:optest} with a time-varying $\hat{\theta}_t$, and an approximate power flow solution $z_t$, satisfies
\begin{equation*}\begin{array}{rl}
     & \normsz{u_{t+1}-u^*(\hat{\theta}_{t+1})}_2 \\
     \leq & \normsz{u_{t}^+-u^*(\hat{\theta}_t)}_2 + \normsz{\underbrace{u_{t+1}}_{[u_{t}^+]^{\bar{u}_{t}}_{\munderbar{u}_{t}}}-u_t^+}_2 + \normsz{u^*(\hat{\theta}_{t+1})-u^*(\hat{\theta}_t)}_2 \\[-0.6cm]
    \leq & \normsz{\overbrace{T(u_t,z_t,\hat{\theta}_t)}^{u_{t}^+}-u^*(\hat{\theta}_t)}_2 + (1+L_u)\normsz{\hat{\theta}_{t+1}-\hat{\theta}_{t}}_2 \\
    \leq & \normsz{T(u_t,\hat{x}_t,\hat{\theta}_t)-u^*(\hat{\theta}_t)}_2 + (1+L_u)\normsz{\hat{\theta}_{t+1}-\hat{\theta}_{t}}_2 \\
    & + \normsz{T(u_t,z_t,\hat{\theta}_t)-T(u_t,\hat{x}_t,\hat{\theta}_t)}_2 \\
    \leq & \xi \normsz{u_{t}-u^*(\hat{\theta}_t)}_2 + (1+L_u)\normsz{\hat{\theta}_{t+1}-\hat{\theta}_{t}}_2 + L_T\normsz{\hat{x}_t-z_t}_2 \\
    \leq & \xi^{t+1} \normsz{u_{0}-u^*(\hat{\theta}_0)}_2 \\
    & + \frac{\max(L_T,1+L_u)}{1-\xi} \sup_{k \leq t} \big( \normsz{\hat{\theta}_{k+1}-\hat{\theta}_{k}}_2 + \normsz{\hat{x}_k-z_k}_2 \big)
\end{array}
\end{equation*}
Consider a ball $\mathcal{B}_{R_u}(u^*(\theta))=\{u | \normsz{u-u^*(\theta)}_2 \leq R_u\}$ in the neighborhood of linear convergence of \eqref{eq:oper}. To ensure that $u_{t+1}$ remains in this neighborhood, it must hold that $\xi \normsz{u_{t}-u^*(\hat{\theta}_t)}_2 + L_u\normsz{\hat{\theta}_{t+1}-\hat{\theta}_{t}}_2 + L_T\normsz{\hat{x}_t-z_t}_2 \leq 
\xi R_u + \allowbreak (1+L_u)\normsz{\hat{\theta}_{t+1}-\hat{\theta}_{t}}_2 + L_T\normsz{\hat{x}_t-z_t}_2 \leq R_u$, or equivalently $$\underbrace{\tfrac{(1+L_u)}{(1-\xi)R_u}}_{\frac{1}{R_\theta}}\normsz{\hat{\theta}_{t+1}-\hat{\theta}_{t}}_2 + \underbrace{\tfrac{L_x}{(1-\xi)R_u}}_{\frac{1}{R_x}}\normsz{\hat{x}_t-z_t}_2 \leq 1.$$

\subsection{Thm.~\ref{thm:opconvglb}: Global convergence and stability}
For a constant $\hat{\theta}$ and using the exact power flow solution $\chi(u,d)$ we have: Since $\normsz{T(u,\chi(u,d),\theta)-u}_2$ is continuous positive definite with respect to $\mathcal{U}^s(\hat{\theta})=\mathcal{U}^*(\hat{\theta})$, \eqref{eq:descdir} implies that $u$ converges to $\mathcal{U}^*(\hat{\theta})$ by the invariance principle for discrete-time systems \cite{la1976stability}. Hence, since the elements in $\mathcal{U}^*(\hat{\theta})$ are asymptotically stable, $\mathcal{U}^*(\hat{\theta})$ is globally asymptotically stable. Then, if we show that the update of $u_t$ in Algorithm~\ref{alg:optest} satisfies the limit property \cite{gao2000equivalent,tran2015equivalences}, we can conclude that $\mathcal{U}^*(\hat{\theta})$ is input-to-state-stable \cite{jiang2001input} by \cite[Thm.~7]{gao2000equivalent}, i.e., there exists $\mathcal{KL},\mathcal{K}$-functions, $\beta(\cdot),\gamma(\cdot)$ respectively, such that
\begin{equation*}\begin{array}{l}
      \distset{u_t}_{\mathcal{U}^*(\hat{\theta}_t)} \\
      \leq \beta(\distset{u_0}_{\mathcal{U}^s(\hat{\theta}_0)},t) + \gamma(\sup_{k < t} \normsz{\hat{\theta}_{k+1}-\hat{\theta}_{k}}_2+\normsz{\hat{x}_k-z_k}_2).
\end{array}
\end{equation*}

To prove the limit property, we extend the $\mathcal{K}$-functions $\alpha_{1,\theta}(\cdot)$ and $\alpha_{2,\theta}(\cdot)$ linearly to get the $\mathcal{K_\infty}$-functions
\begin{equation*}
    \tilde{\alpha}_{i,\theta}(s) = \left\lbrace \begin{array}{ll}
        \alpha_{i,\theta}(s) & \text{if } s \leq r(\theta) \\
        \frac{\alpha_{i,\theta}(r(\theta))}{r(\theta)}s & \text{if } s > r(\theta),
    \end{array}\right.
\end{equation*}

Define $\alpha_1 (s) = \min_{\theta \in \Theta} \tilde{\alpha}_{1,\theta}(s)$, and $\alpha_2 (s) = \max_{\theta \in \Theta} \tilde{\alpha}_{2,\theta}(s)$.
Then, for $u \in \mathcal{U},\theta \in \Theta$ we have 
\begin{equation}\label{eq:Kfunc}\begin{array}{c}
    \alpha_{1}(\distset{u}_{\mathcal{U}^s(\theta)}) \leq  \normsz{T(u,\chi(u,d),\theta)-u}_2^2 \leq  \alpha_{2}(\distset{u}_{\mathcal{U}^s(\theta)}),  
\end{array}
\end{equation}
and $\alpha_i(\cdot)$ are $\mathcal{K}_\infty$-functions by the following lemma: 

\begin{lemma}\label{lemma:Kcont}
If ${\alpha}_{\theta}(\cdot)$ are $\mathcal{K}_\infty$-functions for all $\theta \in \Theta$, and $\Theta$ is a compact set, then $\munderbar{\alpha}(s)=\min_{\theta \in \Theta} {\alpha}_{\theta}(s)$ and $\bar{\alpha}(s)=\max_{\theta \in \Theta} {\alpha}_{\theta}(s)$ are $\mathcal{K}_\infty$-functions.
\end{lemma}

\begin{proof}
First, we prove the result for $\munderbar{\alpha}(\cdot)$. Note that $\munderbar{\alpha}(0)=\min_{\theta \in \Theta} \alpha_{\theta}(0) = 0$. It is also strictly increasing: consider $s_0<s_1$, and $\theta_i = \arg\min_{\theta \in \Theta} \alpha_{\theta}(s_i)$, which are well-defined because $\Theta$ is compact, then $\munderbar{\alpha}(s_0) = \alpha_{\theta_0}(s_0) \leq \alpha_{\theta_1} (s_0) < \alpha_{\theta_1}(s_1) = \munderbar{\alpha}(s_1)$. And $\munderbar{\alpha}(\cdot)$ is radially unbounded, because all $\alpha_\theta(\cdot)$ are. %since $\frac{\alpha_{i,\theta}(r(\theta))}{r(\theta)}>0$ for all $\theta \in \Theta$.

Next, assume that $\munderbar{\alpha}(\cdot)$ is not continuous at some $s_0$. Hence there exists $\bar{\epsilon}>0$ such that $\forall \bar{\delta}>0$ there exists $s_1$ satisfying $\abs{s_0-s_1} < \bar{\delta}$, such that $\abs{\munderbar{\alpha}(s_0)-\munderbar{\alpha}(s_1)} \geq \bar{\epsilon}$. 

Since $\alpha_\theta(\cdot)$ is continuous for any fixed $\theta$, $\forall \epsilon$ there exists $\delta(\epsilon,\theta)$ such that $\abs{s_0-s_1} < \delta(\epsilon,\theta) \implies \abs{{\alpha}_{\theta}(s_0)-{\alpha}_\theta(s_1)} < {\epsilon}$. Take $\epsilon = \frac{\bar{\epsilon}}{2}$; $\bar{\delta}=\min_{\theta \in \Theta} \delta(\frac{\bar{\epsilon}}{2},\theta)>0$, which is positive because $\Theta$ is compact; and $s_1$ such that $\abs{s_0-s_1} < \bar{\delta}$, $\abs{\munderbar{\alpha}(s_0)-\munderbar{\alpha}(s_1)} \geq \bar{\epsilon}$. Assume that $\munderbar{\alpha}(s_0) \leq \munderbar{\alpha}(s_1) - \bar{\epsilon}$, the other case is analogous. Then, $\alpha_{\theta_0}(s_1) \stackrel{\text{cont.} \, \alpha_{\theta}}{<} \alpha_{\theta_0}(s_0) + \frac{\bar{\epsilon}}{2} \stackrel{\text{def.} \, \theta_0}{=} \munderbar{\alpha}(s_0) + \frac{\bar{\epsilon}}{2} \stackrel{\text{disc.} \, \munderbar{\alpha}}{\leq} 
\munderbar{\alpha}(s_1) - \frac{\bar{\epsilon}}{2} \stackrel{\text{def.} \, \theta_1}{=} {\alpha}_{\theta_1}(s_1) - \frac{\bar{\epsilon}}{2}  \stackrel{\theta_1 \, \min}{\leq} 
\alpha_{\theta_0}(s_1) - \frac{\bar{\epsilon}}{2}$, which is a contradiction.

Note that $\bar{\alpha}(\cdot)$ is also strictly increasing: consider $s_0<s_1$, and $\bar{\theta}_i = \arg\max_{\theta \in \Theta} \alpha_{\theta}(s_i)$, then $\bar{\alpha}(s_1) = \alpha_{\bar{\theta}_1}(s_1) \geq \alpha_{\bar{\theta}_0} (s_1) > \alpha_{\bar{\theta}_0}(s_0) = \bar{\alpha}(s_0)$. The rest is analogous to $\munderbar{\alpha}(\cdot)$.
\end{proof}

Since $f(\cdot),g(\cdot)$ are continuously differentiable over compact sets, they are Lipschitz continuous. Since $\max(\cdot,0)$ is also Lipschitz continuous, then $V(\cdot)$ is Lipschitz continuous with constant $L_{V}$, and we have
% \begin{equation*}\label{eq:ISSstab}\begin{array}{rl}
% & V(u_{t}^+,\hat{\theta}_{t})-V(u_{t-1}^+,\hat{\theta}_{t-1})  \\[0.1cm]

% = & V\big(T(u_t ,\hat{x}_t,\hat{\theta}_{t}),\hat{\theta}_{t}\big)-V(u_t ,\hat{\theta}_{t})  \\
%          & + V(u_t ,\hat{\theta}_{t})-V(u_{t-1}^+,\hat{\theta}_{t-1})  \\
%          & + V\big(\underbrace{T(u_t ,z_t,\hat{\theta}_{t})}_{u_t^+},\hat{\theta}_{t}\big) - V\big({T(u_t ,\hat{x}_t,\hat{\theta}_{t})},\hat{\theta}_{t}\big) \\[-0.2cm]

% \stackrel{\eqref{eq:descdir}}{\leq} & -\nu \normsz{T(u_t,\chi(u_t,d_t),\hat{\theta}_t)-u_t}_2^2 \\
% & + L_{V} \normsz{\hat{\theta}_{t}-\hat{\theta}_{t-1}}_2  + L_{V}L_T \normsz{\hat{x}_t-z_t}_2 \\[0.1cm]

% \stackrel{\eqref{eq:Kfunc}}{\leq} & -\nu \alpha_1(\distset{u_t}_{\mathcal{U}^*(\hat{\theta}_t)}) + L_{V} \normsz{\hat{\theta}_{t}-\hat{\theta}_{t-1}}_2 +  L_{V}L_T \normsz{\hat{x}_t-z_t}_2,
% \end{array}
% \end{equation*}
% so that $V(u_{t}^+,\hat{\theta}_{t})-V(u_{0}^+,\hat{\theta}_{0}) \leq - t \nu \alpha_1(\inf_{k\leq t}\distset{u_k}_{\mathcal{U}^*(\hat{\theta}_k)}) + L_{V}\sum_{k=1}^t  \normsz{\hat{\theta}_{k}-\hat{\theta}_{k-1}}_2 + L_T \normsz{\hat{x}_k-z_k}_2 $. Since $\mathcal{U},{\Theta}$ are compact, $V(u,\hat{\theta})$ is lower bounded, hence it must hold that $\inf_{t\geq 1} \distset{u_t}_{\mathcal{U}^s(\hat{\theta}_t)} \leq \sup_{t\geq 1} \alpha^{-1}\big(\frac{L_V}{\nu} (\normsz{\hat{\theta}_{t}-\hat{\theta}_{t-1}}_2+ L_T \normsz{\hat{x}_t-z_t}_2 ) \big)$, i.e., the limit property holds \cite{gao2000equivalent,tran2015equivalences}. 

\begin{equation*}\label{eq:ISSstab}\begin{array}{rl}
& V(u_{t+1},\hat{\theta}_{t+1})-V(u_{t},\hat{\theta}_{t})  \\[0.1cm]

= & V\big(T(u_t ,\hat{x}_t,\hat{\theta}_{t}),\hat{\theta}_{t}\big)-V(u_t ,\hat{\theta}_{t})  \\
& + V(u_{t+1} ,\hat{\theta}_{t+1})-V(u_{t}^+,\hat{\theta}_{t})  \\
         & + V\big(\underbrace{T(u_t ,z_t,\hat{\theta}_{t})}_{u_t^+},\hat{\theta}_{t}\big) - V\big({T(u_t ,\underbrace{\chi(u_t,\hat{d}_t)}_{\hat{x}_t},\hat{\theta}_{t})},\hat{\theta}_{t}\big) \\[-0.2cm]

\stackrel{\eqref{eq:descdir}}{\leq} & -\nu \normsz{T(u_t,\hat{x}_t,\hat{\theta}_t)-u_t}_2^2 \\
& + L_{V} \normsz{\hat{\theta}_{t}-\hat{\theta}_{t-1}}_2  + L_{V}L_T \normsz{\hat{x}_t-z_t}_2 \\[0.1cm]

\stackrel{\eqref{eq:Kfunc}}{\leq} & -\nu \alpha_1(\distset{u_t}_{\mathcal{U}^*(\hat{\theta}_t)}) + L_{V} \normsz{\hat{\theta}_{t+1}-\hat{\theta}_{t}}_2 +  L_{V}L_T \normsz{\hat{x}_t-z_t}_2,
\end{array}
\end{equation*}
so that 
\begin{equation*}\begin{array}{rl}
& V(u_{t+1},\hat{\theta}_{t+1})  -V(u_{0},\hat{\theta}_{0}) \\
\leq & - (t+1) \nu \alpha_1(\inf_{k\leq t}\distset{u_k}_{\mathcal{U}^*(\hat{\theta}_k)}) \\
& +  L_{V} (t+1) \sup_{k\leq t} (\normsz{\hat{\theta}_{k+1}-\hat{\theta}_{k}}_2+ L_T \normsz{\hat{x}_k-z_k}_2 ).
\end{array}
\end{equation*}
Since $\mathcal{U},{\Theta}$ are compact, $V(u,\hat{\theta})$ is lower bounded, hence it must hold that $\inf_t \distset{u_t}_{\mathcal{U}^s(\hat{\theta}_t)} \leq \sup_{t} \alpha^{-1}\big(\frac{L_V}{\nu} (\normsz{\hat{\theta}_{t+1}-\hat{\theta}_{t}}_2+ L_T \normsz{\hat{x}_t-z_t}_2 ) \big)$, i.e., the limit property holds \cite{gao2000equivalent,tran2015equivalences}. 

%\fi

% that's all folks
\end{document}